%% file: arxiv-version.tex
\documentclass[11pt]{article}
\usepackage[margin=1in]{geometry}
\usepackage{fullpage}
\usepackage{graphicx}
\usepackage[utf8]{inputenc}
\usepackage[T1]{fontenc}
\usepackage{hyperref}
\usepackage{url}
\usepackage{booktabs}
\usepackage{amsfonts}
\usepackage{nicefrac}
\usepackage{microtype}
\usepackage[dvipsnames]{xcolor}
\usepackage{multirow}
\usepackage{algorithm}
\usepackage{algorithmicx}
\usepackage[noend]{algpseudocode}
\usepackage{amsmath}
\usepackage{amssymb}
\usepackage{mathtools}
\usepackage{amsthm}
\usepackage[capitalize,noabbrev]{cleveref}
\usepackage{subcaption}
\usepackage{numprint}
\npdecimalsign{.}
\usepackage{lineno}
\usepackage{enumitem}
\usepackage{csquotes}
\usepackage{tcolorbox}

\theoremstyle{plain}
\newtheorem{theorem}{Theorem}[section]

\newtheorem{lemma}[theorem]{Lemma}

\theoremstyle{definition}
\newtheorem{definition}[theorem]{Definition}

\theoremstyle{remark}

\input{math_commands}

\allowdisplaybreaks
\algnewcommand{\LeftComment}[1]{\Statex $\triangleright$ #1}
\algnewcommand\algorithmicforeach{\textbf{for each}}
\algdef{S}[FOR]{ForEach}[1]{\algorithmicforeach\ #1\ \algorithmicdo}
\algnewcommand{\IIf}[1]{\State\algorithmicif\ #1\ \algorithmicthen}
\algnewcommand{\ElseIIf}[1]{\algorithmicelse\ #1}
\algnewcommand{\EndIIf}{\unskip\ \algorithmicend\ \algorithmicif}

\newcommand{\poly}{\operatorname{poly}}

\DeclareRobustCommand{\CM}{{\bf CM}}
\DeclareRobustCommand{\opt}{\mathrm{OPT}}
\definecolor{offwhite}{rgb}{0.98, 0.98, 0.98}

\title{\Large Learning-Augmented Streaming Algorithms for Approximating MAX-CUT}

\author{Yinhao Dong\thanks{School of Computer Science and Technology, University of Science and Technology of China. Email: \mbox{yhdong@mail.ustc.edu.cn}} \and Pan Peng\thanks{School of Computer Science and Technology, University of Science and Technology of China. Email: {ppeng@ustc.edu.cn}} \and Ali Vakilian\thanks{Toyota Technological Institute at Chicago (TTIC). Email: {vakilian@ttic.edu}}}

\date{}

\begin{document}

\maketitle
\begin{abstract}
We study learning-augmented streaming algorithms for estimating the value of MAX-CUT in a graph. In the classical streaming model, while a $1/2$-approximation for estimating the value of MAX-CUT can be trivially achieved with $O(1)$ words of space, Kapralov and Krachun [STOC’19] showed that this is essentially the best possible: for any $\epsilon > 0$, any (randomized) single-pass streaming algorithm that achieves an approximation ratio of at least $1/2 + \epsilon$ requires $\Omega(n / 2^{\text{poly}(1/\epsilon)})$ space.

We show that it is possible to surpass the $1/2$-approximation barrier using just $O(1)$ words of space by leveraging a (machine learned) oracle. Specifically, we consider streaming algorithms that are equipped with an $\epsilon$-accurate oracle that for each vertex in the graph, returns its correct label in $\{-1, +1\}$, corresponding to an optimal MAX-CUT solution in the graph, with some probability $1/2 + \epsilon$, and the incorrect label otherwise. 

Within this framework, we present a single-pass algorithm that approximates the value of MAX-CUT to within a factor of $1/2 + \Omega(\epsilon^2)$ with probability at least $2/3$ for insertion-only streams, using only $\text{poly}(1/\epsilon)$ words of space. We also extend our algorithm to fully dynamic streams while maintaining a space complexity of $\poly(1/\eps,\log n)$ words. 
\end{abstract}

\section{Introduction}

Given an undirected, unweighted graph, the MAX-CUT problem seeks to find a partition of the vertices (a cut) that maximizes the number of edges with endpoints on opposite sides of the partition (such edges are said to be ``cut'' by the partition). MAX-CUT is a well-known NP-hard problem. The best-known approximation algorithm, developed by Goemans and Williamson \cite{goemans1995improved}, achieves a $0.878$-approximation ratio, which is the best possible under the Unique Games Conjecture \cite{khot2007optimal}.

We consider the problem of estimating the value of MAX-CUT in a graph under the streaming model of computation, where MAX-CUT denotes some fixed optimal solution, and the value of a solution is defined as the number of edges in the graph that are cut by the partition. In this model, the graph is presented as a sequence of edge insertions and deletions, known as a graph stream. The objective is to analyze the graph's structure using minimal space. The model is called insertion-only if the stream contains only edge insertions; if both insertions and deletions are allowed, it is referred to as the dynamic model.

In the streaming model, a $1/2$-approximation for estimating the value of MAX-CUT can be trivially achieved using $O(1)$ words of space, where a word represents the space required to encode the size of the graph. This is done by simply counting the total number $m$ of edges in the graph and outputting $m/2$. However, Kapralov and Krachun \cite{KK19} showed that this is essentially the best possible: for any $\epsilon > 0$, any (randomized) single-pass streaming algorithm in the insertion-only model that achieves an approximation ratio of at least $1/2 + \epsilon$ requires $\Omega(n / 2^{\text{poly}(1/\epsilon)})$ space.

We focus on the problem of estimating the value of MAX-CUT in a streaming setting within the framework of \emph{learning-augmented algorithms}. These algorithms utilize predictions from a machine learning model to solve a problem, where the predictions typically include some information about the optimal solution, future events, or yet unread data in the input stream  (see, e.g., \cite{hsu2019learning,jiang2020learning,chen2022triangle}). Learning-augmented algorithms have gained significant attention recently, partly due to numerous breakthroughs in machine learning. Ideally, such algorithms should be both \emph{robust} and \emph{consistent}: when the predictions are accurate, the algorithm should outperform the best-known classical algorithms, and when the predictions are inaccurate, the algorithm should still provide performance guarantees that are close to or match those of the best-known classical algorithms. {Despite the extensive research in the area of learning-augmented algorithms over the past few years\footnote{See \url{https://algorithms-with-predictions.github.io/} for an up-to-date repository of publications on this topic.}, our understanding of this framework within the streaming model remains limited.}

In this work, we consider the \emph{noisy prediction} model, also referred to as {\em $\eps$-accurate predictions}, where the algorithm has oracle access to a prediction vector $ Y \in \{-1, 1\}^n $. Each entry $ Y_v $ is \emph{independently correct} with probability $ \frac{1}{2} + \epsilon $, where $ \epsilon \in (0, \frac{1}{2}] $ represents the \emph{bias} of the predictions. 
Specifically, for each $ v \in V $, we have $ \Pr[Y_v = x_v^*] = \frac{1}{2} + \epsilon $ and $ \Pr[Y_v = -x_v^*] = \frac{1}{2} - \epsilon $, 
where $ x^* \in \{-1, 1\}^n $ denotes some fixed but unknown optimal solution for MAX-CUT. This model captures a scenario where an ML algorithm (or oracle) provides predictions for the values of $ x^* $ that are unreliable and noisy, being only slightly better than random guesses (i.e., just above the $ 1/2 $ probability that a random solution would satisfy). Recently, variants of this prediction model have been used to design improved approximation algorithms for MAX-CUT and constraint satisfaction problems (CSPs) \cite{cohen2024max,ghoshal2025constraint}.

Specifically, we study the following question:

\begin{center}
\begin{tcolorbox}[colback=offwhite, colframe=lightgray, boxrule=1.5pt, arc=5mm, width=\textwidth]
{\em Given an oracle $\mathcal{O}$ that provides $\eps$-accurate predictions $Y$ about the optimal \textup{MAX-CUT} of a streaming graph $G$, can we improve upon the worst-case approximation ratio and space trade-off established by \cite{KK19}, specifically the $(\frac{1}{2} + \epsilon)$ ratio with $\Omega(n)$ space complexity, for estimating the \textup{MAX-CUT} value?}
\end{tcolorbox}
\end{center}

\subsection{Our Results}
We provide an affirmative answer by presenting a single-pass  streaming algorithm that surpasses the $1/2$-approximation barrier for the value of MAX-CUT using only $\text{poly}(1/\epsilon) $ (resp., $\text{poly}(1/\epsilon,\log n) $)words of space in insertion-only (resp., fully dynamic) streams. Formally, we establish the following result: 

\begin{theorem}\label{thm:main_introduction}
Let $ \epsilon \in (0, \frac{1}{2}] $.  
Given oracle access to an $ \epsilon $-accurate predictor, there exists a single-pass streaming algorithm that provides a $ \left( \frac{1}{2} + \Omega(\epsilon^2) \right) $-approximation for estimating the MAX-CUT value of a graph in insertion-only (resp., fully dynamic) streams using $\poly(1/\eps)$ (resp., $\poly(1/\epsilon, \log n)$) words of space with probability at least $2/3$. 
\end{theorem}
By using median trick, we can boost the success probability to $1-\delta$, where $\delta \in (0,1)$, with an additional $\log (1/\delta)$ factor  in  space complexity.

We remark that the ``robustness'' of our learning-augmented algorithm comes for free, as we can always compute the number of edges in the graph and report half of them as a $\frac12$-approximation solution, regardless of the value of $\epsilon$. Furthermore, with predictions, it gains an advantage of $\Omega(\eps^2)$ in the approximation ratio. Furthermore, our algorithm does not require the exact value of $\eps$; a constant-factor approximation of $\eps$ is sufficient. This estimation is only needed to determine the sample size.

Our algorithm is based on the observation that when the maximum degree of a graph is relatively small (i.e. smaller than $\poly(\eps)\cdot m$), the number of edges with endpoints having different predicted labels already provides a strictly better-than-$\frac{1}{2}$ approximation of the MAX-CUT value. For general graphs, we employ some well-known techniques, such as the CountMin sketch and $\ell_0$-sampling, to distinguish between high-degree and low-degree vertices in both insertion-only and dynamic streams. By combining a natural greedy algorithm with an algorithm tailored for the low-degree part, we achieve a non-trivial approximation. The space complexity of the resulting algorithm is primarily determined by the subroutines for identifying the set of high-degree vertices, which can be bounded by $\text{poly}(1/\eps,\log n)$.

We note that it is standard to assume that the noisy oracle $\mathcal{O}$ has a compact representation and that its space complexity is not included in the space usage of our algorithm, following the conventions in streaming learning-augmented algorithms~\cite{hsu2019learning,jiang2020learning,chen2022triangle,aamand2023improved}. Indeed, as noted in~\cite{hsu2019learning}, a reliable oracle can often be learned space-efficiently in practice as well.
Moreover, we show that in the case of random-order streams, our algorithm only needs to query $\mathcal{O}$ for the labels of a constant number of vertices (see \cref{sec:constantquery}).

Additionally, our algorithm actually works in a weaker model. That is, we can assume that in our streaming framework, the predicted label of each vertex is associated with its edges. Thus, the elements in a stream can be represented as $ (e = (u, v), Y_u, Y_v)$, where the predictions $(Y_u)_{u\in V}$ remain consistent throughout the stream. The case of dynamic streams is defined analogously. Notably, in this model, no additional space is required to store the predictors.

\subsection{Related Work}
Learning-augmented algorithms have been actively researched in online algorithms \cite{mahdian2007allocating,kumar2018improving,angelopoulos2020online,antoniadis2020secretary,bamas2020primal,im2021online,lykouris2021competitive,antoniadis2023online,lattanzi2020online}, data structures \cite{mitzenmacher2018model,ferragina2020the,VaidyaKMK21,lin2022learning,sato2023fast}, graph algorithms \cite{dinitz2021faster,chen2022faster,banerjee2023graph,lattanzi2023speeding,davies2023predictive,brand2024dynamic,liu2023predicted,henzinger2024complexity,depavia2024learning,cohen2024max,ghoshal2025constraint,braverman2024learning}, and sublinear-time algorithms \cite{indyk2019learning,eden2021learning,li2023learning,schiefer2023learned}. Our algorithms fit into the category of learning-augmented streaming algorithms. Hsu et al.~\cite{hsu2019learning} introduced learning-augmented streaming algorithms for frequency estimation, and Jiang et al.~\cite{jiang2020learning} extended this framework to various norm estimation problems in data stream. Recently, Aamand et al.~\cite{aamand2019frequency,aamand2023improved} developed learning-augmented frequency estimation algorithms, that improve upon the work of~\cite{hsu2019learning}. 
Additionally, Chen et al. \cite{chen2022triangle} studied single-pass streaming algorithms to estimate the number of triangles and four-cycles in a graph. It is worth mentioning that both our work and previous efforts on learning-augmented streaming algorithms mainly focus on using predictors to improve the corresponding space-accuracy trade-offs.
Furthermore, recent studies have explored learning-augmented algorithms for MAX-CUT and constraint satisfaction problems (CSPs)~\cite{cohen2024max,ghoshal2025constraint} within a variant of our prediction model. In particular, for the Max-CUT problem, Cohen-Addad et al.~\cite{cohen2024max} achieved a $(0.878+\tilde{\Omega}(\eps^4))$-approximation using SDP. However, it is important to emphasize that their setting differs significantly from ours: they focus on finding the MAX-CUT in the offline setting, whereas our goal is to estimate its size in the streaming setting. Additionally, Ghoshal et al.~\cite{ghoshal2025constraint} developed an algorithm that yields near-optimal solutions when the average degree is larger than a threshold determined only by the bias of the predictions, independent of the instance size. They further extended this result to weighted graphs and the MAX-$3$LIN problem.

\medskip
The aforementioned lower bound proven in \cite{KK19} is, in fact, the culmination of a series of works exploring the trade-offs between space and approximation ratio for the streaming MAX-CUT value problem, including \cite{kapralov2014streaming,kogan2015sketching,kapralov20171+}.
In our main context, we focus on the problem of estimating the value of MAX-CUT in a streaming setting. Another related problem is that of finding the MAX-CUT itself in a streaming model. Since even outputting the MAX-CUT requires $\Omega(n)$ space, research in this area primarily considers streaming algorithms that use $\tilde{O}(n)$  space, known as the semi-streaming model, where $\tilde{O}(\cdot)$ hides polylogarithmic factors. Beyond the trivial $\frac{1}{2}$-approximation algorithm, which simply partitions the vertices randomly into two parts, there exists a $(1 - \epsilon)$-approximation algorithm for finding MAX-CUT. This algorithm works by constructing a $(1 - \epsilon)$-cut or spectral sparsifier of the input streaming graph, which can be achieved in $\tilde{O}(n)$ space (see e.g. \cite{ahn2012graph,kapralov2020fast}), and then finding the MAX-CUT of the sparsified graph (which may require exponential time).

\medskip
Approximating MAX-CUT in the offline setting has also been extensively studied in other data-driven approaches. These methods aim to learn the configuration of an algorithm best suited for instances drawn from a given distribution. For further details on this line of research, see~\cite{balcan2017learning,balcan2020data} and the references therein.

\section{Preliminaries and Problem Statement}

\paragraph{Notations.} 
Throughout the paper we let $G=(V,E)$ be an undirected, unweighted graph with $n$ vertices and $m$ edges.
Given a vertex $v\in V$ and disjoint sets $S, T \subseteq V$, $e(v,S)$ denotes the number of edges incident to $v$ whose other endpoint belongs to $S$; $e(v,S) = |\{(u,v) \;|\; (u,v)\in E \text{ and } u\in S\}|$. Similarly, we define $e(S,T)$ to denote the number of edges with one endpoint in $S$ and the other endpoint in $T$; $e(S, T) = \sum_{v\in T} e(v, S)$.
A cut $(S, T)$ of $G$ is a bipartition of $V$ into two disjoint subsets $S$ and $T$. The value of the cut $(S, T)$ is $e(S,T)$.

In the descriptions of our algorithms, for a vertex set $S \subseteq V$, we define $S^+$ (resp. $S^-$) as the set of vertices with predicted $+$ (resp. $-$) signs by the given $\eps$-accurate oracle $\mathcal{O}$. Note that $S^+$ and $S^-$ form a bipartition of $S$. Throughout the paper, all space complexities are measured in words unless otherwise specified. The space complexity in bits is larger by a factor of $\log m + \log n$.

In our algorithms, we utilize several well-known techniques in the streaming model, which we define below.

\paragraph{CountMin Sketch~\cite{cormode2005improved}.}
Given a stream of $m$ integers from $[n]$, let $f_i$ denote number of occurrences of $i$ in the stream, for any $i\in [n]$.
In CountMin sketch, there are $k$ uniform and independent random hash functions $h_1,h_2,\dots,h_k:[n] \rightarrow [w]$ and an array $C$ of size $k \times w$. The algorithm maintains $C$, such that $C[\ell,s] = \sum_{j:h_\ell(j)=s}f_j$ at the end of the stream. Upon each query of an arbitrary $i\in [n]$, the algorithm returns $\tilde{f}_i = \min_{\ell \in [k]} C[\ell,h_\ell(i)]$.

\begin{theorem}[CountMin Sketch~\cite{cormode2005improved}]
\label{thm:CM}
    For any $i\in [n]$ and any $\ell \in [k]$, we have $\tilde{f}_i \geq f_i$ and $\E[C[\ell,h_\ell(i)]] \leq f_i + \frac{m}{w}$. The space complexity is $O(kw)$ words. 
    Let $\eps,\delta \in (0,1)$.
    If we set $k=\lceil \frac{e}{\eps}\rceil$ and $w = \lceil \ln \frac{1}{\delta}\rceil$, then for any $i\in [n]$, we have $\tilde{f}_i \geq f_i$, and with probability at least $1-\delta$, $\tilde{f}_i \leq f_i + \eps m$. The corresponding space complexity is $O(\frac{1}{\eps}\ln \frac{1}{\delta})$ words.
\end{theorem}

\paragraph{Reservoir Sampling~\cite{vitter1985random}.}
The algorithm samples $k$ elements from a stream, such that at any time $m \geq k$, the sample consists of a uniformly random subset of size $k$ of the elements seen so far.
The space complexity of the algorithm is $O(k(\log n + \log m))$ bits.

\begin{definition}[$\ell_0$-sampler~\cite{jowhari2011tight}]
\label{def:l0-sampler}
Let $x\in \sR^n$ be a non-zero vector and $\delta\in (0,1)$. 
An $\ell_0$-sampler for $x$ returns FAIL with probability at most $\delta$ and otherwise returns some index $i$ such that $x_i \neq 0$ with probability $\frac{1}{|\operatorname{supp}(x)|}$ where $\operatorname{supp}(x)=\{i\mid x_i\neq 0\}$ is the support of $x$.
\end{definition}

The following theorem states that $\ell_0$-samplers can be maintained using a single pass in dynamic streams.
\begin{theorem}[Theorem~2~in~\cite{jowhari2011tight}]
\label{thm:l0-sampler}
Let $\delta \in (0,1)$.
There exists a single-pass streaming algorithm for maintaining an $\ell_0$-sampler for a non-zero vector $x\in \sR^n$ (with failure probability $\delta$) in the dynamic model using $O(\log^2 n \log \delta^{-1})$ bits of space.
\end{theorem}

\paragraph{Problem Statement.}

In this work, we consider the problem of estimating the value of MAX-CUT of a graph $G$ in the learning-augmented streaming model. 
Formally, given an undirected, unweighted graph $G=(V,E)$, which is represented as a sequence of edges, i.e., $\sigma := \langle e_1,e_2,\dots \rangle$ for $e_i \in E$, our goal is to scan the sequence in one pass and output an estimate of the MAX-CUT value,
while minimizing space usage. When the sequence contains only edge insertions, it is referred to as an \emph{insertion-only} stream. When the sequence contains both edge insertions and deletions, it is referred to as a \emph{dynamic} stream. Note that in dynamic streams, the sequence of the stream is often represented as $\sigma := \langle (e_1,\Delta_1),(e_2,\Delta_2),\dots \rangle$, where for each $i$, $e_i\in E$ and $\Delta_i=1$ (resp. $-1$) denotes edge insertion (resp. deletion).

\medskip
Furthermore, we assume that the algorithms are equipped with an oracle $\mathcal{O}$, which provides $\eps$-accurate predictions $Y\in \{-1,1\}^n$ where $\eps \in (0,\frac{1}{2}]$. This information is provided via an external oracle, and for the purpose of our algorithm, we assume that the edges in the stream are annotated with the predictions of their endpoints.  

Specifically, for each vertex $v\in V$, $ \Pr[Y_v = x_v^*] = \frac{1}{2} + \epsilon $ and $ \Pr[Y_v = -x_v^*] = \frac{1}{2} - \epsilon $, where $ x^* \in \{-1, 1\}^n $ is some fixed but unknown optimal solution. Following previous works on MAX-CUT with $\epsilon$-accurate predictions~\cite{cohen2024max,ghoshal2025constraint}, we also assume that the $(Y_v)_{v\in V}$ are independent. 

\medskip
Finally, since we assume that accuracy parameter of predictions, $\epsilon$, is known to the algorithm up to a constant factor in advance, and the space complexity of all algorithms in the paper is at most $\poly(\log n, 1/\epsilon,1/\delta)$, we will assume throughout that $m=\Omega(\eps^{-11}\delta^{-7})$. 
Otherwise, when $m=O(\epsilon^{-11}\delta^{-7})$, we can simply store all edges in $\poly(1/\epsilon, 1/\delta)$ space and compute the MAX-CUT exactly.

\section{The Simple Case of Low-Degree Graphs}\label{sec:low-degree}

In this section, we show that when all vertices have ``low-degree'' then following the $\epsilon$-accurate predictions, we can beat the $\frac{1}{2}$-approximation barrier for streaming MAX-CUT. Informally, once an edge 
$(u,v)$ arrives in the stream, we consult the noisy oracle $\mathcal{O}$ and if their endpoints are predicted to be in different parts by the noisy oracle, i.e., $\mathcal{O}(u) \neq \mathcal{O}(v)$, we increase the size of the cut by one. This can be simply implemented using $O(\log n)$ bits of space. Refer to~\cref{alg:low-deg} for a formal description. 

\begin{algorithm}[ht]
\begin{algorithmic}[1]
\Require Graph $G$ with $\Delta < \frac{\eps^2\delta m}{4} - \frac{1}{4\eps^2}$ as a stream of edges, an $\eps$-accurate oracle $\mathcal{O} \rightarrow \{-1, 1\}$
\Ensure Estimate of the MAX-CUT value of $G$
\State {\bf initialize} $X \leftarrow 0$.
\ForEach{edge $(u,v)$ in the stream}
    \State $Y_u = \mathcal{O}(u), Y_v = \mathcal{O}(v)$
    \If{$Y_u \neq Y_v$}
        $X \leftarrow X+1$
    \EndIf
\EndFor
\State \textbf{return} $X$
\end{algorithmic}
\caption{A learning-augmented streaming algorithm for estimating MAX-CUT value in low-degree graphs}
\label{alg:low-deg}
\end{algorithm}

\medskip
Next, we analyze the space complexity and approximation guarantee of~\cref{alg:low-deg}.

\begin{theorem}
\label{thm:low-degree-estimate}
Let $\eps \in (0, \frac{1}{2}]$ and $\delta \in (0,1)$. Given an $\epsilon$-accurate oracle $\mathcal{O}$ for the MAX-CUT of $G$, if the maximum degree $\Delta$ in $G$ satisfies $\Delta < \frac{\eps^2 \delta m}{4} - \frac{1}{4\eps^2}$, then with probability at least $1-\delta$, \cref{alg:low-deg} outputs a $(\frac{1}{2} + \eps^2)$-approximation of the MAX-CUT value of $G$. The algorithm uses $O(1)$ words of space.
\end{theorem}
\begin{proof}
We first analyze the space complexity.
In \cref{alg:low-deg}, we maintain a counter $X$ for the edges whose endpoints have different signs, which takes $O(\log n)$ bits of space. 
Therefore, the space complexity of the algorithm is $O(1)$ words.

Next, we show that \cref{alg:low-deg} with high probability outputs a $(\frac{1}{2} + \eps^2)$-approximation of the MAX-CUT when all vertices have low degrees.

Since $(V^+,V^-)$ is a feasible cut of $G$, where $V^+$ (resp. $V^-$) is the set of vertices with predicted $+$ (resp. $-$) signs by the given $\eps$-accurate oracle $\mathcal{O}$,
we have $X \leq \opt$, where $\opt$ denotes the (optimal) MAX-CUT value of $G$. 
For each edge $(u,v)\in E$, define an indicator random variable $X_{uv}$ such that $X_{uv} = 1$ if $Y_u \neq Y_v$, and is zero otherwise. 

\medskip
Let $x^*$ be the assignment vector corresponding to the MAX-CUT in $G$. Specifically, if $(S, V \setminus S)$ is the MAX-CUT of $G$, then $x^*_v = 1$ if $v \in S$, and $-1$ otherwise. Consider the following two cases:
    \begin{enumerate}[label=\textsf{(I.\arabic*)},leftmargin = *]
        \item\label{eq:predicted-cut-edge} If $x_u^* \neq x_v^*$, then $\Pr[Y_u \neq Y_v]=\Pr[Y_u = x_u^* \wedge Y_v = x_v^*] + \Pr[Y_u \neq x_u^* \wedge Y_v \neq x_v^*] = \frac{1}{2} + 2\eps^2$. 
        \item\label{eq:predicted-noncut-edge}  If $x_u^* = x_v^*$, then $\Pr[Y_u \neq Y_v]=\Pr[Y_u = x_u^* \wedge Y_v \neq x_v^*] + \Pr[Y_u \neq x_u^* \wedge Y_v = x_v^*] = \frac{1}{2} - 2\eps^2$.
    \end{enumerate}
    Note that $X=\sum_{(u,v)\in E}X_{uv}$ and $\E[X_{uv}] = \Pr[Y_u \neq Y_v]$. Then, by the linearity of expectation,
    \begin{align*}
        \E[X]
        = \sum_{\substack{(u,v)\in E\\ x_u^* \neq x_v^*}}\E[X_{uv}] + \sum_{\substack{(u,v)\in E\\ x_u^* = x_v^*}}\E[X_{uv}] 
        &= \opt \cdot (\frac{1}{2}+2\eps^2)  + (m-\opt) \cdot (\frac{1}{2}-2\eps^2) \\
        &= (\frac{1}{2} - 2\eps^2) m + 4\eps^2 \cdot \opt \\
        &\geq (\frac{1}{2} + 2\eps^2) \cdot \opt,
    \end{align*}
    where the last inequality holds since $\frac{1}{2} - 2\eps^2 \geq 0$ for $\eps \in (0,\frac{1}{2}]$ and $m \geq \mathrm{OPT}$.

    \medskip
    Next, we bound $\Var[X]$. Note that $\Var[X_{uv}] = \E[X_{uv}] - \E[X_{uv}]^2$ since $(X_{uv})_{(u,v)\in E}$ are indicator random variables. 
    By~\ref{eq:predicted-cut-edge} and~\ref{eq:predicted-noncut-edge}, $\Var[X_{uv}] =\frac{1}{4} - 4\eps^4$.
    \begin{enumerate}[label=\textsf{(II.\arabic*)},leftmargin =*]
        \item If $x_u^* \neq x_v^*$, then $\Var[X_{uv}] =\left( \frac{1}{2}+2\eps^2\right) - \left( \frac{1}{2}+2\eps^2\right)^2 = \frac{1}{4} - 4\eps^4$.
        \item  If $x_u^* = x_v^*$, then $\Var[X_{uv}] =\left( \frac{1}{2}-2\eps^2\right) - \left( \frac{1}{2}-2\eps^2\right)^2 = \frac{1}{4} - 4\eps^4$.
    \end{enumerate}

    \medskip
    For any pair of distinct edges $(u,v), (u,w)\in E$, we compute $\E[X_{uv}X_{uw}] = \Pr[X_{uv} = X_{uw} = 1] = \Pr[Y_v = Y_w = - Y_u]$ and $\Cov[X_{uv}, X_{uw}] = \E[X_{uv}X_{uw}] - \E[X_{uv}] \cdot  \E[X_{uw}]$. 
    \begin{enumerate}[label=\textsf{(III.\arabic*)},leftmargin = *]
        \item If $x_u^* = x_v^* \text{ and } x_u^* = x_w^*$, then 
        \begin{align*}
            \E[X_{uv}X_{uw}] 
            &= \Pr[Y_u = x_u^* \wedge Y_v \neq x_v^* \wedge Y_w \neq x_w^*] + \Pr[Y_u \neq x_u^* \wedge Y_v = x_v^* \wedge Y_w = x_w^*] \\
            &= \left( \frac{1}{2}+\eps\right)\left( \frac{1}{2}-\eps\right)^2 + \left( \frac{1}{2}-\eps\right)\left( \frac{1}{2}+\eps\right)^2 = \frac{1}{4} - \eps^2.
        \end{align*}
        
        Therefore, $\Cov[X_{uv}, X_{uw}] =  \left( \frac{1}{4} - \eps^2\right) - \left(\frac{1}{2} - 2\eps^2\right)^2 = \eps^2-4\eps^4$.
        
        \item If $x_u^* = x_v^* \text{ and } x_u^* \neq x_w^*$, then 
        \begin{align*}
            \E[X_{uv}X_{uw}] 
            &= \Pr[Y_u = x_u^* \wedge Y_v \neq x_v^* \wedge Y_w = x_w^*] + \Pr[Y_u \neq x_u^* \wedge Y_v = x_v^* \wedge Y_w \neq x_w^*] \\
            &= \left( \frac{1}{2}+\eps\right)^2\left( \frac{1}{2}-\eps\right) + \left( \frac{1}{2}-\eps\right)^2\left( \frac{1}{2}+\eps\right) = \frac{1}{4} - \eps^2.  
        \end{align*}
        
        Therefore, $\Cov[X_{uv}, X_{uw}] =  \left( \frac{1}{4} - \eps^2\right) - \left(\frac{1}{2} - 2\eps^2\right)\left(\frac{1}{2} + 2\eps^2\right) = 4\eps^4-\eps^2$.
        
        \item If $x_u^* \neq x_v^* \text{ and } x_u^* = x_w^*$, then 
        \begin{align*}
            \E[X_{uv}X_{uw}] 
            &= \Pr[Y_u = x_u^* \wedge Y_v = x_v^* \wedge Y_w \neq x_w^*] + \Pr[Y_u \neq x_u^* \wedge Y_v \neq x_v^* \wedge Y_w = x_w^*] \\
            &= \left(\frac{1}{2}+\eps\right)^2\left( \frac{1}{2}-\eps\right) + \left( \frac{1}{2}-\eps\right)^2\left( \frac{1}{2}+\eps\right) = \frac{1}{4} - \eps^2.
        \end{align*}
        Therefore, $\Cov[X_{uv}, X_{uw}] =  \left( \frac{1}{4} - \eps^2\right) - \left(\frac{1}{2} + 2\eps^2\right)\left(\frac{1}{2} - 2\eps^2\right) = 4\eps^4-\eps^2$
        
        \item If $x_u^* \neq x_v^* \text{ and } x_u^* \neq x_w^*$, then 
        \begin{align*}
            \E[X_{uv}X_{uw}] 
            &= \Pr[Y_u = x_u^* \wedge Y_v = x_v^* \wedge Y_w = x_w^*] + \Pr[Y_u \neq x_u^* \wedge Y_v \neq x_v^* \wedge Y_w \neq x_w^*] \\
            &= \left( \frac{1}{2}+\eps\right)^3 + \left( \frac{1}{2}-\eps\right)^3 = \frac{1}{4} +3\eps^2. 
        \end{align*}
        
        Therefore, $\Cov[X_{uv}, X_{uw}] = \left( \frac{1}{4} + 3\eps^2\right) - \left(\frac{1}{2} + 2\eps^2\right)^2 = \eps^2-4\eps^4$.
    \end{enumerate}
    
    Hence, $\Cov[X_{uv}, X_{uw}] \leq \eps^2-4\eps^4$. So,  
    \begin{align*}
        \Var[X] &= \sum_{(u,v)\in E} \Var[X_{uv}] + \sum_{\substack{(u,v), (u,w)\in E\\ v\neq w}}\Cov[X_{uv}, X_{uw}] + \sum_{\substack{(u,v), (w,z)\in E\\ u,v,w,z \text{ all distinct}}}\Cov[X_{uv}, X_{wz}]\\
        &\leq m\cdot \left( \frac{1}{4}-4\eps^4 \right) + m\cdot \Delta \cdot \left(\eps^2-4\eps^4\right) + 0 \leq \left( \frac{1}{4}+\Delta \eps^2\right)m.
    \end{align*}
    Then, by applying Chebyshev's inequality,
    \begin{align*}
        \Pr\left[X \leq \left(\frac{1}{2}+\eps^2\right) \opt\right] 
        \leq \Pr\left[|X-\E[X]| \geq  \eps^2 \cdot \opt \right]
        &\leq \frac{\Var[X]}{\eps^4 \cdot {\opt}^2} \\
        &\leq \frac{\left( \frac{1}{4}+\Delta \eps^2\right)m}{\eps^4 \cdot {\opt}^2} &&\rhd \Var[X] \le ( \frac{1}{4}+\Delta \eps^2)m\\
        &\leq \frac{4\cdot(\frac{1}{4}+\Delta \eps^2)}{\eps^4 m} &&\rhd \opt\geq \frac{m}{2}\\
        &< \delta. &&\rhd \Delta < \frac{\eps^2 \delta m}{4} - \frac{1}{4\eps^2}
    \end{align*}
    So, with probability at least $1-\delta$, $X \geq (\frac{1}{2}+\eps^2)\cdot \opt$.
\end{proof}

\section{Our Algorithm for General Graphs}

In this section, we present our $(\frac{1}{2}+\Omega(\eps^2))$-approximation 
for streaming MAX-CUT in $O(\poly(1/\eps,1/\delta))$ space using $\eps$-accurate predictions. For better presenting our algorithmic ideas, we first consider the problem in the {\em random order} streams in Section~\ref{sec:random-order}. Next, in Section~\ref{sec:arbitrary-order}, we show how to extend our algorithm to arbitrary order streams using more advanced sketching techniques. Finally, in Section~\ref{sec:dynamic}, we show that our algorithm also work in dynamic streams where both edge insertions and deletions are allowed.

\paragraph{Offline Implementation.} We first describe our algorithm in the offline setting. 
Recall that our algorithm from Section~\ref{sec:low-degree} is effective when the maximum degree in the graph is smaller than a pre-specified threshold $\phi = \Theta(\eps^2 \delta m)$, where $\delta$ is the target failing probability of the algorithm. 
Let $H$ and $L$ respectively denote the high-degree (i.e., vertices with degree at least $\phi$) and low-degree (i.e., vertices with degree less than $\phi$) vertices in the input graph $G = (V,E)$.
Now, by the guarantee of~\cref{thm:low-degree-estimate}, suppose we have a $(\frac{1}{2}+\eps^2)$-approximation of MAX-CUT on the induced subgraph $G[L]$ denoted by $(L^+, L^-)$. This cut is exactly the cut suggested by the given $\eps$-accurate oracle $\mathcal{O}$. 
Next, we extend this cut using the vertices in $H$ in a greedy manner. We pick vertices in $H$ one by one in an arbitrary order and each time add them to either $L^+$ or $L^-$ sections so that the size of cut maximized. We denote the resulting cut by $(C, V\setminus C)$. Another cut we consider is simply $(H, L)$. We show at least one these two cuts is a $(\frac{1}{2}+\Omega(\eps^2))$-approximation for MAX-CUT on $G$ with high probability.

Throughout the paper, we define high-degree vertices as those with degree $\geq \frac{\eps^2 m}{c}$, where $c=\frac{80}{\delta}$. Vertices with degrees below this threshold are considered low-degree.

\begin{theorem}
\label{thm:best-approx}
    Let $\eps \in (0,\frac{1}{2}]$ and $\delta \in (0,1)$. 
    The best of $(C, V\setminus C)$ and $(H, L)$ cuts is a $(\frac{1}{2}+\frac{\eps^2}{16})$-approximation for MAX-CUT on $G$ with probability at least $1-\delta$.
\end{theorem}

Without loss of generality, we assume that $e(H,L) < (\frac{1}{2}+\eps^2)\cdot \opt$, otherwise we are done.
Then it suffices to show that $e(C,V\setminus C) \geq (\frac{1}{2}+\frac{\eps^2}{16})\cdot \opt$. Recall that the cut $(C,V\setminus C)$ is obtained by extending the cut $(L^+,L^-)$ of $G[L]$ suggested by the given $\eps$-accurate oracle using the vertices in $H$ in a greedy manner. We first show that \cref{alg:low-deg} works for the induced subgraph $G[L]$.
\begin{lemma}
\label{lem:threshold-low-degree-graph}
    Let $\eps \in (0,\frac{1}{2}]$ and $\delta \in (0,1)$.
    Suppose that $e(H,L) < (\frac{1}{2}+\eps^2)\cdot \opt$, 
    $c=\frac{80}{\delta}$
    and $m>\frac{c^3 \delta}{\eps^4}+\frac{1}{4\eps^4}$. Then there exists a $(\frac{1}{2} + \eps^2)$-approximation of the MAX-CUT value of $G[L]$ with probability at least $1-\delta$.
\end{lemma}
\begin{proof}
    Note that $m = m_L + m_H + e(H,L)$, where $m_L$ (resp. $m_H$) is the number of edges in $G[L]$ (resp. $G[H]$).
    Then we have $m_L = m - m_H - e(H,L) > m - \frac{4c^2}{\eps^4} - (\frac{1}{2}+\eps^2)m$, since $|H|\leq \frac{2m}{\eps^2 m / c} = \frac{2c}{\eps^2}$ and $e(H,L) < (\frac{1}{2}+\eps^2)\cdot \opt \leq (\frac{1}{2}+\eps^2)m$.

    Since the high-degree threshold is $\frac{\eps^2 m}{c}$, we have $\Delta_L < \frac{\eps^2 m}{c}$, where $\Delta_L$ is the maximum degree of $G[L]$.
    It is easy to check that $\frac{\eps^2 m}{c} < \frac{\eps^2 \delta}{4}(m - \frac{4c^2}{\eps^4} - (\frac{1}{2}+\eps^2)m) - \frac{1}{4\eps^2}$ by substituting in the conditions for $c$ and $m$. It follows that $\Delta_L < \frac{\eps^2\delta m_L}{4} - \frac{1}{4\eps^2}$.

    Therefore, by the guarantee of~\cref{thm:low-degree-estimate}, with probability at least $1-\delta$, there exists a $(\frac{1}{2}+\eps^2)$-approximation of MAX-CUT on $G[L]$ denoted by $(L^+, L^-)$, i.e., $e(L^+,L^-) \geq (\frac{1}{2}+\eps^2) \cdot \opt_L$, where $\opt_L$ is the size of the MAX-CUT value of $G[L]$.
\end{proof}

\begin{lemma}
\label{lem:lower-bound-opt-L}
    Let $\eps \in (0,\frac{1}{2}]$.
    Suppose that $e(H,L) = \alpha \cdot \opt$ (where $\alpha < \frac{1}{2}+\eps^2$) and $m > \frac{8c^2}{\eps^8}$. We have $\opt_L > (1-\alpha - \eps^4) \cdot \opt$, where $\opt_L$ is the size of the MAX-CUT value of $G[L]$.
\end{lemma}
\begin{proof}
    Note that $\opt \leq \opt_L + \opt_H + e(H,L)$, where $\opt_L$ (resp. $\opt_H$) is the size of the MAX-CUT of $G[L]$ (resp. $G[H]$). It follows that
    \begin{align*}
        \opt_L &\geq \opt - \opt_H - e(H,L)\\
        &\geq \opt - \frac{4c^2}{\eps^4} - e(H,L) &&\rhd \opt_H \leq m_H \leq |H|^2 \leq \frac{4c^2}{\eps^4}\\
        &= \opt - \frac{4c^2}{\eps^4} - \alpha \cdot \opt &&\rhd e(H,L) =\alpha \cdot \opt\\
        &> \opt - \eps^4 \cdot \opt - \alpha \cdot \opt &&\rhd \frac{4c^2}{\eps^4} < \eps^4 \cdot \frac{m}{2} \leq \eps^4 \cdot \opt \\
        &= (1-\alpha-\eps^4)\cdot \opt.
    \end{align*}
\end{proof}

\begin{proof}[Proof of \cref{thm:best-approx}]
    Since the algorithm returns the best of two cuts, if the value of the cut $(H, L)$ is at least $(\frac{1}{2}+\eps^2)\cdot \opt$ (i.e., $e(H,L) \geq (\frac{1}{2}+\eps^2)\cdot \opt$), then we are done.
    Hence, without loss of generality, we can assume that the size of the cut $(H, L)$ is strictly less than $(\frac{1}{2}+\eps^2)\cdot \opt$.
    Based on \cref{lem:threshold-low-degree-graph} and \cref{lem:lower-bound-opt-L}, we have
    \begin{align*}
        e(C,V\setminus C) &\geq e(L^+,L^-) + \sum_{v\in H}\max\{e(v,L^+),e(v,L^-)\}\\
        &\geq \left(\frac{1}{2}+\eps^2\right) \opt_L + \sum_{v\in H}\max\{e(v,L^+),e(v,L^-)\} &&\rhd e(L^+,L^-) \geq (\frac{1}{2}+\eps^2) \opt_L\\
        &\geq \left(\frac{1}{2}+\eps^2\right) \opt_L+ \frac{1}{2}\cdot e(H,L) &&\rhd \max\{e(v,L^+),e(v,L^-)\} \geq \frac{e(v,L)}{2} \\
        &> \left(\frac{1}{2}+\eps^2\right)(1-\alpha-\eps^4) \opt + \frac{\alpha}{2}\cdot \opt &&\rhd \opt_L > (1-\alpha-\eps^4) \opt \\
        &= \left( \frac{1}{2} + (1-\alpha)\eps^2 - \frac{\eps^4}{2}-\eps^6\right) \opt\\
        &>  \left( \frac{1}{2} + \left(\frac{1}{2}-\eps^2\right)\eps^2 - \frac{\eps^4}{2}-\eps^6\right) \opt &&\rhd \alpha < \frac{1}{2} + \eps^2\\
        &\geq \left( \frac{1}{2}+\frac{\eps^2}{16}\right) \opt, &&\rhd \left(\frac{1}{2}-\eps^2\right)\eps^2 - \frac{\eps^4}{2}-\eps^6 \geq \frac{\eps^2}{16}
    \end{align*}
    with probability at least $1-\delta$.
\end{proof}

\subsection{Warm-up: Random Order Streams}\label{sec:random-order}

\paragraph{Overview of the Algorithm.} Suppose that the edges of $G = (V, E)$ arrive one by one in a random order stream. 
At a high level, we would like to run our algorithm from Section~\ref{sec:low-degree} for low-degree vertices $L$ (i.e., the induced subgraph on $L$), and then, using a greedy approach, add the high-degree vertices $H$ to the constructed cut. 
However, the set $H$ and $L$ are not known a priori in the stream, and it is not clear how to store the required information to run the described algorithm, space efficiently.

\medskip
To detect high-degree vertices and collect sufficient information to run the greedy approach on them at the end of the stream, we rely on the random order of the stream. We store a small number of edges from the beginning of the stream and then gather degree information for those vertices that are candidates for being high-degree.
More precisely, we store the first $\poly(1/\eps, 1/\delta)$ edges in the (random order) stream and use them to identify a set $\tilde{H}$ of size  $\poly(1/\eps,1/\delta)$ that contains all high-degree vertices $H$ with probability $1-\delta$. We then store all edges between vertices in $\tilde{H}$ throughout the stream, which requires $\poly(1/\epsilon,1/\delta)$ words of space. Additionally, for every vertex $v \in \tilde{H}$, we maintain the number of edges between $v$ and $\tilde{L}^+ \coloneqq V^+ \setminus \tilde{H}$, as well as between $v$ and $\tilde{L}^- \coloneqq V^- \setminus \tilde{H}$. 
It is straightforward to check that these counters require $\poly(1/\eps,1/\delta)$ words of space in total. We then apply \Cref{alg:low-deg} to the graph $G[V\setminus \tilde{H}]$ to approximate the MAX-CUT value of $G[V\setminus \tilde{H}]$.

\medskip
Finally, at the end of stream, since we can compute the degree of all vertices in $\tilde{H}$ exactly, we can determine the set of high degree vertices $H\subseteq \tilde{H}$. First, for the remaining vertices $\tilde{S} \coloneqq \tilde{H}\setminus H$, we (hypothetically) feed them to our algorithm for the low-degree vertices. Note that we have all edges between any pair of vertices in $\tilde{S}$, as well as all the number of their incident edges to $\tilde{L}^+$ and $\tilde{L}^-$. Therefore, we can compute the size of $(L^+, L^-)$ cut exactly. Now it only remains to run the greedy algorithm for the high-degree vertices $H$. Similarly, since we have computed the number of incident edges of high-degree vertices to $\tilde{L}^+, \tilde{L}^-$, and we have stored all edges with both endpoints in $\tilde{H}$, we can perform the greedy extension of $(\tilde{L}^+\cup \tilde{S}^+, \tilde{L}^-\cup \tilde{S}^-)$ using $H$. Moreover, using the same set of degree information, we can compute the size of $(H, V\setminus H)$ as well. Then, we can return the best of these two cuts as our estimate of MAX-CUT in $G$.

\medskip
Next, we prove the approximation guarantee and the space complexity of the algorithm. The algorithm is formally given in Algorithm~\ref{alg:rand-order}.

\begin{algorithm}[h!]
\begin{algorithmic}[1]
\Require Graph $G$ as a random-order stream of edges, an $\eps$-accurate oracle $\mathcal{O} \rightarrow \{-1, 1\}$, 
a high-degree threshold $\theta = \frac{\eps^2 m}{c}$ (where $c=\frac{80}{\delta}$).
\Ensure The estimate of the MAX-CUT value of $G$.
\LeftComment{\textbf{Preprocessing phase}}
\State {\bf initialize} $F \leftarrow \emptyset$, $\tilde{H} \leftarrow \emptyset$, $e(L^+,L^-) \leftarrow 0$. 

\LeftComment{\textbf{Streaming phase}}

\ForEach{edge $(u,v)$ in the first $\frac{\beta}{\delta^3\eps^4}$ edges of the stream (where $\beta$ is a sufficiently large universal constant)}
\State $F \leftarrow F \cup \{(u,v)\}$, $\tilde{H} \leftarrow \tilde{H} \cup \{u,v\}$ \phantomsection \label{line:store-first}

\EndFor

\ForEach{vertex $v\in \tilde{H}$}
\State {\bf initialize} $e(v,V\setminus \tilde{H})\leftarrow 0$, $e(v,L^+)\leftarrow 0$, $e(v,L^-)\leftarrow 0$.
\EndFor
\ForEach{remaining edge $(u,v)$ in the stream}
\State $Y_u = \mathcal{O}(u), Y_v = \mathcal{O}(v)$

\If{$u \in \tilde{H}$ and $v \in \tilde{H}$} 
 $F \leftarrow F \cup \{(u,v)\}$ \label{line:store-remaining}
\Else
\IIf{$u \in \tilde{H}$ and $v\in V\setminus \tilde{H}$} $e(u,V\setminus \tilde{H}) \leftarrow e(u,V\setminus \tilde{H})+1$
\IIf{$v \in \tilde{H}$ and $u\in V\setminus \tilde{H}$} $e(v,V\setminus \tilde{H}) \leftarrow e(v,V\setminus \tilde{H})+1$
\EndIf

\If{$u\in \tilde{H}$ and $v\in V\setminus \tilde{H}$}
\IIf{$Y_v = 1$} $e(u,L^+) \leftarrow e(u,L^+) + 1$ \ElseIIf $e(u,L^-) \leftarrow e(u,L^-) + 1$
\EndIf
\If{$u\in V\setminus \tilde{H}$ and $v\in V\setminus \tilde{H}$ and $Y_u \neq Y_v$} 
 \State $e(L^+,L^-) \leftarrow e(L^+,L^-) + 1$ (i.e., apply \cref{alg:low-deg} on $G[V\setminus \tilde{H}]$) \label{line:apply-in-streams}
\EndIf
\EndFor
\LeftComment{\textbf{Postprocessing phase}}
\State $H \leftarrow \{v\in \tilde{H}: |\{e \in F :v\in e\}| + e(v,V\setminus \tilde{H}) \geq \theta\}$, $L \leftarrow V \setminus H$
\ForEach{edge $(u,v)\in F$}
\State $Y_u = \mathcal{O}(u), Y_v = \mathcal{O}(v)$
\If{$u\in H$ and $v \in L$}
\IIf{$Y_v = 1$} $e(u,L^+) \leftarrow e(u,L^+) + 1$ \ElseIIf $e(u,L^-) \leftarrow e(u,L^-) + 1$
\EndIf
 \IIf{$u\in \tilde{H}\setminus H$ and $v\in \tilde{H}\setminus H$ and $Y_u \neq Y_v$} $e(L^+,L^-) \leftarrow e(L^+,L^-) + 1$ \label{line:apply-in-post}

\EndFor
\ForEach{vertex $v\in \tilde{H}\setminus H$}
\State $Y_v = \mathcal{O}(v)$
\IIf{$Y_v = 1$} $e(L^+,L^-) \leftarrow e(L^+,L^-) + e(v,L^-)$ \ElseIIf $e(L^+,L^-) \leftarrow e(L^+,L^-) + e(v,L^+)$ \label{line:apply-intra}
\EndFor
\State $\mathrm{ALG}_1 \leftarrow e(L^+,L^-) + \sum_{v\in H} \max\{e(v,L^-),e(v,L^+)\}$
\State $\mathrm{ALG}_2 \leftarrow \sum_{v\in H} (e(v,L^-) + e(v,L^+))$
\State \textbf{return} $\max\{\mathrm{ALG}_1,\mathrm{ALG}_2\}$
\end{algorithmic}
\caption{Estimating the MAX-CUT value in (insertion-only) random-order streams}
\label{alg:rand-order}
\end{algorithm}

\begin{theorem}
\label{thm:random-order}
Let $\eps\in (0,\frac{1}{2}]$ and $\delta \in (0,1)$.
Given an $\epsilon$-accurate oracle $\mathcal{O}$ for the MAX-CUT of $G$, there exists a single-pass $(\frac{1}{2} + \frac{\eps^2}{16})$-approximation algorithm for estimating the MAX-CUT value of $G$ in the insertion-only random-order streams. The algorithm uses $O(\frac{1}{\delta^6\eps^8})$ words of space. 
The approximation holds with probability at least $1-\delta$.
\end{theorem}

\begin{lemma}
\label{lem:high-deg-sample}
Let $\delta \in (0,1)$.
Then $\tilde{H}$ contains all high-degree vertices in $G$ with probability at least $1-\frac{\delta}{2}$.
\end{lemma}
\begin{proof}
    Suppose that $v$ is a high-degree vertex in $G$, i.e., $\deg(v) \geq \frac{\eps^2 m}{c}$. 
    For each edge $e_i$ incident to $v$, define an indicator random variable 
    \begin{equation*}
            X_{i}= \begin{cases}1, & \text { if } e_i \text{ is in the first } \frac{\beta}{\delta^3\eps^4} \text{ edges of the random order stream} \\ 0, & \text { otherwise }\end{cases}
    \end{equation*}
    So, $\E[X_i] = \Pr[X_i=1] = \frac{\beta}{\delta^3\eps^4 m}$. We define $X:= \sum_{i\in [\deg(v)]} X_i$ to denote the number of edges incident to $v$ that appear in the first $\frac{\beta}{\delta^3\eps^4}$ edges of the random order stream. Then, $\E[X] = \sum_{i\in [\deg(v)]} \E[X_i] = \frac{\beta\deg(v)}{\delta^3\eps^4 m}$. 

    Next, we bound $\Var[X]$. For any $i\in [\deg(v)]$, $\Var[X_i] = \E[X^2_i] - (\E[X_i])^2 = \E[X_i] - (\E[X_i])^2 = \frac{\beta}{\delta^3\eps^4 m} - (\frac{\beta}{\delta^3\eps^4 m})^2 = \frac{\beta}{\delta^3\eps^4 m}-\frac{\beta^2}{\delta^6\eps^8 m^2}$.
    Since $(X_i)_{i\in [\deg(v)]}$ are negatively correlated, $\Var[X] \le \sum_{i\in [\deg(v)]}\Var[X_i] = \deg(v)\cdot \left(\frac{\beta}{\delta^3\eps^4 m}-\frac{\beta^2}{\delta^6\eps^8 m^2}\right) \leq \frac{\beta\deg(v)}{\delta^3\eps^4 m}$.
    Therefore, by Chebyshev's inequality, 
    \begin{align*}
        \Pr[v\notin \tilde{H}] = \Pr[X\leq 0] 
        = \Pr\left[X \leq \E[X]- \frac{\beta\deg(v)}{\delta^3\eps^4 m}\right]
        &\leq \frac{\Var[X]\cdot \delta^6\eps^{8} m^2}{\beta^2\deg^2(v)}\\
        &\leq \frac{\delta^3\eps^4 m}{\beta\cdot \deg(v)} &&\rhd \Var[X] \leq \frac{\beta\deg(v)}{\delta^3\eps^4 m}\\
        &\leq \frac{\delta^3 \eps^2 c}{\beta}. &&\rhd \deg(v)\geq \frac{\eps^2 m}{c}
    \end{align*}
    By our definition of high-degree vertices, there are at most $\frac{2m}{\eps^2 m / c} =\frac{2c}{\eps^2}$ high-degree vertices in $G$.
    By union bound, with probability at least $1-\frac{2c}{\eps^2}\cdot \frac{\delta^3 \eps^2 c}{\beta} = 1-\frac{2c^2\delta^3}{\beta} \geq 1-\frac{\delta}{2}$ (since $c=\frac{80}{\delta}$ and $\beta$ is sufficiently large), $\tilde{H}$ contains all high-degree vertices in $G$.
\end{proof}

\begin{proof}[Proof of \cref{thm:random-order}]
    In \cref{alg:rand-order}, we store the first $\frac{\beta}{\delta^3\eps^4}$ edges of the stream (\cref{line:store-first}) and the remaining edges with both endpoints in $\tilde{H}$ (\cref{line:store-remaining}). We also maintain several counters for vertices in $\tilde{H}$. 
    Therefore, the total space complexity of \cref{alg:rand-order} is $O(\frac{1}{\delta^6\eps^8})$ words.
    
    Next, we show the approximation guarantee of \cref{alg:rand-order}. 
    Note that at the end of the stream we can identify $H$ exactly, using the information stored during the streaming phase.
    In \cref{alg:rand-order}, we apply \cref{alg:low-deg} to the subgraph $G[V\setminus \tilde{H}]$ during the streaming phase (Line~\ref{line:apply-in-streams}). Hypothetically, during the postprocessing phase, we apply \cref{alg:low-deg} to the subgraph $G[\tilde{H}\setminus H]$ (Line~\ref{line:apply-in-post}) and to edges with one endpoint in $\tilde{H} \setminus H$ and the other endpoint in $V\setminus \tilde{H}$ (\cref{line:apply-intra}). 
    This is equivalent to applying \cref{alg:low-deg} to $G[L]$ since $L=(V\setminus \tilde{H}) \cup (\tilde{H}\setminus H)$.
    Then the approximation guarantee of \cref{alg:rand-order} follows directly from \cref{thm:best-approx} (with failure probability $\frac{\delta}{2}$) and \cref{lem:high-deg-sample} by using union bound.
\end{proof}

\subsection{Arbitrary Order Streams}\label{sec:arbitrary-order}
\paragraph{Overview of the Algorithm.} 

Unlike random order streams, where we can identify a set $\tilde{H}$, containing the high-degree vertices $H$, by storing only a small number of edges from the beginning of the stream, finding high-degree vertices is more challenging in arbitrary order streams. 
To handle this, we employ {\em reservoir sampling}~\cite{vitter1985random}. 
Specifically, we uniformly sample $\poly(1/\eps,1/\delta)$ edges from the stream. Then, at the end of the stream, we use these sampled edges to compute a small set $\tilde{H}$ such that $\tilde{H} \supseteq H$. This approach is similar to what we used in the random order stream; however, in this case, we can only retrieve $\tilde{H}$ rather than $H$ at the end of the stream.

We also need to estimate the number of incident edges related to high-degree vertices. To this end,
we use techniques from vector sketching, particularly those developed for {\em frequency estimation} and {\em heavy hitters}. We consider the sketching techniques for heavy hitters, specifically the randomized summaries of CountMin sketch~\cite{cormode2005improved}, corresponding to $V^+$ and $V^-$, denoted by $\CM[V^+]$ and $\CM[V^-]$, respectively. Intuitively, by the end of the stream, $\CM[V^+]$ and $\CM[V^-]$ will contain the estimates of the number of incident edges of high-degree vertices to sets $V^+$ and $V^-$, respectively. 
More precisely, as an edge $(u,v)$ arrive, we increment the counters related to $u$ in $\CM[V^{\mathcal{O}(v)}]$, and increment the counters related to $v$ in $\CM[V^{\mathcal{O}(u)}]$, where $\mathcal{O}(v)$ and $\mathcal{O}(u)$ are respectively the predicted sign of $v$, $u$ by the given $\eps$-accurate oracle $\mathcal{O}$. 

\medskip
Note that we cannot detect $H$ exactly by the end of stream. Instead, we have $\tilde{H}$. We use $\CM[V^+]$ and $\CM[V^-]$ to approximately compute the value of the cut $(\tilde{L}^+, \tilde{L}^-)$ on the induced subgraph $G[\tilde{L}]$, where $\tilde{L}:= V\setminus \tilde{H}$. Then, using $\CM[V^+]$ and $\CM[V^-]$, we can approximately run the greedy approach for $\tilde{H}$. Unlike the random order setting where we could implement the greedy extension exactly, here we can only store the approximate values. However, we show that our estimates of the number of incident edges of high-degree vertices to sets $V^+$ and $V^-$ are only off by $\poly(\epsilon,\delta)\cdot m$ additive terms with high probability, and they suffice to provide a strictly better than $\frac{1}{2}$-approximation for MAX-CUT on $G[\tilde{L}]$ and its extension with $\tilde{H}$. Similarly, we can use $\CM[V^+]$ and $\CM[V^-]$ to compute the size of $(\tilde{H}, \tilde{L})$ approximately too. Then, we can show that for sufficiently small value of $\epsilon$, the best of two candidate cuts is a 
$(\frac{1}{2}+\Omega(\eps^2))$-approximation for MAX-CUT on $G$ with  probability $1-\delta$.

\medskip
Next, we prove the approximation guarantee and the space complexity of our algorithm for arbitrary order streams. The algorithm is formally given in Algorithm~\ref{alg:arbitrary-order}.

\begin{algorithm}[h!]
\begin{algorithmic}[1]
\Require Graph $G$ as an arbitrary-order stream of edges, an $\eps$-accurate oracle $\mathcal{O} \rightarrow \{-1, 1\}$
\Ensure Estimate of the MAX-CUT value of $G$
\LeftComment{\textbf{Preprocessing phase}}
\State Initialize $e(V^+,V^-) \leftarrow 0$.
\State Set $h_1,\dots,h_k:[n] \rightarrow [w]$ be $2$-wise independent hash functions, with $k=\lceil \frac{e}{\eps^{7}\delta^3}\rceil$ and $w = \lceil \ln \frac{8\beta}{\eps^4\delta^4}\rceil$ (where $\beta$ is a sufficiently large universal constant).
\State Initialize $\CM[V^+]$ and $\CM[V^-]$ to zero.
\LeftComment{\textbf{Streaming phase}}
\ForEach{edge $(u,v)$ in the stream}
\State $Y_u = \mathcal{O}(u), Y_v = \mathcal{O}(v)$
\IIf{$Y_u \neq Y_v$} $e(V^+,V^-) \leftarrow e(V^+,V^-) + 1$
\ForEach{$\ell \in [k]$}
\State $\CM[V^{\mathcal{O}(u)}][\ell,h_\ell(v)] \leftarrow \CM[V^{\mathcal{O}(u)}][\ell,h_\ell(v)] + 1$
\State $\CM[V^{\mathcal{O}(v)}][\ell,h_\ell(u)] \leftarrow \CM[V^{\mathcal{O}(v)}][\ell,h_\ell(u)] + 1$
\EndFor
\EndFor

\State In parallel, uniformly sample $\frac{\beta}{\delta^3\eps^4}$ edges $F$ in the stream via reservoir sampling~\cite{vitter1985random}. 

\LeftComment{\textbf{Postprocessing phase}}
\ForEach{$v\in V$}
\State $f_v^+ \leftarrow \min_{\ell\in [k]} \CM[V^+][\ell,h_\ell (v)]$
\State $f_v^- \leftarrow \min_{\ell\in [k]} \CM[V^-][\ell,h_\ell (v)]$
\EndFor
\State $\tilde{H} := \bigcup_{(u,v)\in F} \{u,v\}$
\State $\tilde{H}^+ \leftarrow \{v\in \tilde{H}: Y_v = 1\}, \tilde{H}^- \leftarrow V \setminus \tilde{H}^+$
\State $\tilde{e}(\tilde{L}^+,\tilde{L}^-) \leftarrow e(V^+,V^-)  - \sum_{v\in \tilde{H}^+} f_v^- - \sum_{v\in \tilde{H}^-} f_v^+$
\State $\mathrm{ALG}_1 \leftarrow \tilde{e}(\tilde{L}^+,\tilde{L}^-) + \sum_{v\in \tilde{H}} \max\{f_v^-,f_v^+\}$
\State $\mathrm{ALG}_2 \leftarrow \sum_{v\in \tilde{H}} (f_v^- + f_v^+)$
\State \textbf{return} $\max\{\mathrm{ALG}_1,\mathrm{ALG}_2\}$

\end{algorithmic}
\caption{Estimating the MAX-CUT value in (insertion-only) arbitrary-order streams}
\label{alg:arbitrary-order}
\end{algorithm}

\begin{theorem}
\label{thm:arbitrary-order}
Let $\eps\in (0,\frac{1}{2}]$ and $\delta \in (0,1)$. 
Given an $\epsilon$-accurate oracle $\mathcal{O}$ for the MAX-CUT of $G$, 
there exists a single-pass $(\frac{1}{2} + \Omega(\eps^2))$-approximation algorithm for estimating the MAX-CUT value of $G$ in the insertion-only arbitrary-order streams. The algorithm uses $O( \frac{1}{\eps^{7}\delta^3}\ln\frac{1}{\eps \delta} )$ words of space. The approximation holds with probability at least $1-\delta$.
\end{theorem}

\begin{lemma}
\label{lem:high-deg-reservoir}
Let $\delta \in (0,1)$.
Then
 $\tilde{H}$ contains all high-degree vertices $H$ in $G$ with probability at least $1-\frac{\delta}{2}$.
\end{lemma}
\begin{proof}
    The proof is basically similar to that of \cref{lem:high-deg-sample}.
    Suppose that $v$ is a high-degree vertex in $G$, i.e., $\deg(v) \geq \frac{\eps^2 m}{c}$. 
    For each edge $e_i$ incident to $v$, define an indicator random variable 
    \begin{equation*}
            X_{i}= \begin{cases}1, & \text { if } e_i\in F \\ 0, & \text { otherwise }\end{cases}
    \end{equation*}
    So, $\E[X_i] = \Pr[X_i=1] = \frac{\beta}{\delta^3\eps^4 m}$. We define $X:= \sum_{i\in [\deg(v)]} X_i$ to denote the number of edges incident to $v$ that are sampled by the end of the stream. Then, $\E[X] = \sum_{i\in [\deg(v)]} \E[X_i] = \frac{\beta\deg(v)}{\delta^3\eps^4 m}$. 

    Next, we bound $\Var[X]$. For any $i\in [\deg(v)]$, $\Var[X_i] = \E[X^2_i] - (\E[X_i])^2 = \E[X_i] - (\E[X_i])^2 = \frac{\beta}{\delta^3\eps^4 m} - (\frac{\beta}{\delta^3\eps^4 m})^2 = \frac{\beta}{\delta^3\eps^4 m}-\frac{\beta^2}{\delta^6\eps^8 m^2}$.
    Since $(X_i)_{i\in [\deg(v)]}$ are negatively correlated, $\Var[X] \le \sum_{i\in [\deg(v)]}\Var[X_i] = \deg(v)\cdot \left(\frac{\beta}{\delta^3\eps^4 m}-\frac{\beta^2}{\delta^6\eps^8 m^2}\right) \leq \frac{\beta\deg(v)}{\delta^3\eps^4 m}$.
    Therefore, by Chebyshev's inequality, 
    \begin{align*}
        \Pr[v\notin \tilde{H}] = \Pr[X\leq 0] 
        = \Pr\left[X \leq \E[X]- \frac{\beta\deg(v)}{\delta^3\eps^4 m}\right]
        &\leq \frac{\Var[X]\cdot \delta^6\eps^{8} m^2}{\beta^2\deg^2(v)}\\
        &\leq \frac{\delta^3\eps^4 m}{\beta\cdot \deg(v)} &&\rhd \Var[X] \leq \frac{\beta\deg(v)}{\delta^3\eps^4 m}\\
        &\leq \frac{\delta^3 \eps^2 c}{\beta}. &&\rhd \deg(v)\geq \frac{\eps^2 m}{c}
    \end{align*}
    By our definition of high-degree vertices, there are at most $\frac{2m}{\eps^2 m / c} =\frac{2c}{\eps^2}$ high-degree vertices in $G$.
    By union bound, with probability at least $1-\frac{2c}{\eps^2}\cdot \frac{\delta^3 \eps^2 c}{\beta} = 1-\frac{2c^2\delta^3}{\beta} \geq 1-\frac{\delta}{2}$ (since $c=\frac{80}{\delta}$ and $\beta$ sufficiently large), $\tilde{H}$ contains all high-degree vertices in $G$.
\end{proof}

\begin{lemma}
\label{lem:error-bound-arbitrary}
    Let $\eps \in (0,\frac{1}{2}]$ and $\delta \in (0,1)$.
    For all vertices $v\in \tilde{H}$, (1) $f_v^+ \geq e(v,V^+)$ and $f_v^- \geq e(v,V^-)$. (2) with probability at least $1-\frac{\delta}{4}$, $f_v^+ \leq e(v,V^+) + 2\eps^{7} \delta^3 m$ and $f_v^- \leq e(v,V^-) + 2\eps^{7} \delta^3 m$.
\end{lemma}
\begin{proof}
    (1) By \cref{thm:CM}, we have $f_v^+ \geq e(v,V^+)$ and $f_v^- \geq e(v,V^-)$ for all $v \in V$.

    (2) By \cref{thm:CM}, we have $f_v^+ \leq e(v,V^+) + 2\eps^{7} \delta^3 m$ and $f_v^- \leq e(v,V^-) + 2\eps^{7} \delta^3 m$ with probability at least $1-\frac{\eps^4\delta^4}{8\beta}$ for any $v \in V$. Since $|\tilde{H}|\leq \frac{2\beta}{\delta^3\eps^4}$, by union bound, with probability at least $1-\frac{\eps^4\delta^4}{8\beta} \cdot \frac{2\beta}{\delta^3\eps^4} = 1-\frac{\delta}{4}$, we have $f_v^+ \leq e(v,V^+) + 2\eps^{7} \delta^3 m$ and $f_v^- \leq e(v,V^-) + 2\eps^{7} \delta^3 m$ for all $v\in \tilde{H}$.
\end{proof}

\begin{proof}[Proof of \cref{thm:arbitrary-order}]
    In \cref{alg:arbitrary-order}, we use CountMin sketch with $k=\lceil \frac{e}{\eps^{7} \delta^3}\rceil$ and $w = \lceil \ln \frac{8\beta}{\eps^4\delta^4}\rceil$ to estimate the number of incident edges of high-degree vertices to with respect to $V^+$ and $V^-$, which takes $O(\frac{1}{\eps^{7}\delta^3}\ln\frac{1}{\eps^4 \delta^4}) = O( \frac{1}{\eps^{7}\delta^3}\ln\frac{1}{\eps \delta} )$ words of space, by \cref{thm:CM}. 
    Also, we use reservoir sampling to uniformly sample $O(\frac{1}{\delta^3\eps^4})$ edges, which takes $O(\frac{1}{\delta^3\eps^4})$ words of space. 
    Therefore, the total space complexity of \cref{alg:arbitrary-order} is $O( \frac{1}{\eps^{7}\delta^3}\ln\frac{1}{\eps \delta} )$ words.
    
    Next, we show the approximation guarantee of \cref{alg:arbitrary-order}.
    Recall that in arbitrary order streams, we cannot detect $H$ by the end of stream, and we have $\tilde{H}$ instead.
    Let $A_1 := e(\tilde{L}^+,\tilde{L}^-) + \sum_{v\in \tilde{H}} \max \{e(v,\tilde{L}^-),e(v,\tilde{L}^+) \}$ denote the value of the cut $(\tilde{C},V\setminus \tilde{C})$ obtained by running \cref{alg:low-deg} on $G[\tilde{L}]$ and then assigning the vertices in $\tilde{H}$ in a greedy manner.
    Let $A_2 := \sum_{v\in \tilde{H}} (e(v,\tilde{L}^-) + e(v,\tilde{L}^+))$ denote the value of the cut $(\tilde{H},\tilde{L})$. Suppose that the best of $(\tilde{C},V\setminus \tilde{C})$ and $(\tilde{H},\tilde{L})$ cuts is a $(\frac{1}{2}+\frac{\eps^2}{16})$-approximation for the MAX-CUT value of $G$.
    Recall that $\mathrm{ALG}_1 = \tilde{e}(\tilde{L}^+,\tilde{L}^-) + \sum_{v\in \tilde{H}} \max\{f_v^-,f_v^+\}$ and $\mathrm{ALG}_2 = \sum_{v\in \tilde{H}} (f_v^- + f_v^+)$.
    In the following, we show that $\mathrm{ALG}_1$ and $\mathrm{ALG}_2$ are good approximations of $A_1$ and $A_2$, respectively. Therefore, $\max\{\mathrm{ALG}_1,\mathrm{ALG}_2\}$ returned by \cref{alg:arbitrary-order} is also a strictly better than $\frac{1}{2}$-approximation for the MAX-CUT value of $G$.

    Since $V^+ = \tilde{H}^+ \cup \tilde{L}^+$ and $V^- = \tilde{H}^- \cup \tilde{L}^-$, we have $e(v,V^+) = e(v,\tilde{H}^+) + e(v,\tilde{L}^+)$ and $e(v,V^-) = e(v,\tilde{H}^-) + e(v,\tilde{L}^-)$ for any vertex $v\in V$. By \cref{lem:error-bound-arbitrary}, with probability at least $1-\frac{\delta}{4}$, for all vertices $v\in \tilde{H}$, we have 
    \begin{align*}
        e(v,\tilde{H}^+) + e(v,\tilde{L}^+) \leq f_v^+ \leq e(v,\tilde{H}^+) + e(v,\tilde{L}^+) + 2\eps^{7}\delta^3 m,\\
        e(v,\tilde{H}^-) + e(v,\tilde{L}^-) \leq f_v^- \leq e(v,\tilde{H}^-) + e(v,\tilde{L}^-) + 2\eps^{7}\delta^3m.
    \end{align*}
    
    Since $|\tilde{H}|\leq \frac{2\beta}{\delta^3\eps^4}$,
    we have $0 \leq e(v,\tilde{H}^+),e(v,\tilde{H}^-) \leq \frac{2\beta}{\delta^3\eps^{4}}$. Therefore, we have
    \begin{align*}
        e(v,\tilde{L}^+) \leq f_v^+ \leq \frac{2\beta}{\delta^3\eps^{4}} +2\eps^{7}\delta^3 m+ e(v,\tilde{L}^+),\\
        e(v,\tilde{L}^-) \leq f_v^- \leq \frac{2\beta}{\delta^3\eps^{4}} +2\eps^{7}\delta^3 m + e(v,\tilde{L}^-).
    \end{align*}
    Then we have
    \begin{align*}
        0 \leq f_v^+ - e(v,\tilde{L}^+)\leq \frac{2\beta}{\delta^3\eps^{4}} +2\eps^{7}\delta^3 m && \text{and} && 0 \leq f_v^- - e(v,\tilde{L}^-)\leq \frac{2\beta}{\delta^3\eps^{4}} +2\eps^{7}\delta^3 m,
    \end{align*}
    and
    \begin{align*}
        \max\{e(v,\tilde{L}^+),e(v,\tilde{L}^-) \}  \leq \max\{ f_v^+,f_v^-\} \leq \frac{2\beta}{\delta^3\eps^{4}} + 2\eps^7 
 \delta^3 m + \max\{e(v,\tilde{L}^+),e(v,\tilde{L}^-) \}.
    \end{align*}

    Note that 
    \begin{align*}
        e(\tilde{L}^+,\tilde{L}^-) &= e(V^+,V^-) - e(\tilde{H}^+,\tilde{H}^-) - \sum_{v\in \tilde{H}^+}e(v,\tilde{L}^-) - \sum_{v\in \tilde{H}^-}e(v,\tilde{L}^+),\\
        \tilde{e}(\tilde{L}^+,\tilde{L}^-) &= e(V^+,V^-)  - \sum_{v\in \tilde{H}^+} f_v^- - \sum_{v\in \tilde{H}^-} f_v^+.
    \end{align*}
    We have 
    \begin{align*}
        e(\tilde{L}^+,\tilde{L}^-) - \tilde{e}(\tilde{L}^+,\tilde{L}^-) &= \sum_{v\in \tilde{H}^+} (f_v^- - e(v,\tilde{L}^-)) + \sum_{v\in \tilde{H}^-} (f_v^+ - e(v,\tilde{L}^+)) - e(\tilde{H}^+,\tilde{H}^-)\\
        &\leq \frac{2\beta}{\delta^3\eps^{4}}\cdot \left(\frac{2\beta}{\delta^3\eps^{4}} + 2\eps^7 \delta^3 m\right) + \frac{2\beta}{\delta^3\eps^{4}}\cdot \left(\frac{2\beta}{\delta^3\eps^{4}} + 2\eps^7  \delta^3 m\right) - 0
        = \frac{8\beta^2}{\delta^6\eps^8} + 8\beta \eps^3 m.
    \end{align*}

    Therefore,
    \begin{align*}
        A_1 - \mathrm{ALG}_1 &= (e(\tilde{L}^+,\tilde{L}^-) - \tilde{e}(\tilde{L}^+,\tilde{L}^-)) + \sum_{v\in\tilde{H}} (\max\{e(v,\tilde{L}^-), e(v,\tilde{L}^+)\} - \max\{f_v^-,f_v^+\})\\
        &\leq \left(\frac{8\beta^2}{\delta^6\eps^8} + 8\beta \eps^3 m \right)+ \frac{2\beta}{\delta^3\eps^4} \cdot 0
        = \frac{8\beta^2}{\delta^6\eps^8} + 8\beta \eps^3 m.
    \end{align*}
    So, $\mathrm{ALG}_1 \geq A_1 - (\frac{8\beta^2}{\delta^6\eps^8} + 8\beta \eps^3 m) = A_1 - \Theta(\eps^3 m)$. 

    Similarly,
    \begin{align*}
        A_2 - \mathrm{ALG}_2 = \sum_{v\in \tilde{H}} ((  e(v,\tilde{L}^-) - f_v^-) + (  e(v,\tilde{L}^+) - f_v^+))
        \leq \frac{2\beta}{\delta^3 \eps^4} \cdot (0+0) = 0.
    \end{align*}
    So, $\mathrm{ALG}_2 \geq A_2$. 
    
    Since we assume that $\max\{A_1,A_2\}$ is a $(\frac{1}{2}+\frac{\eps^2}{16})$-approximation for the MAX-CUT value of $G$, we have $\mathrm{ALG}:= \max\{\mathrm{ALG}_1, \mathrm{ALG}_2\} \geq \max\{A_1,A_2\}- \Theta(\eps^3 m)  \geq (\frac{1}{2}+\frac{\eps^2}{16})\cdot \opt - \Theta(\eps^3 \cdot \opt) = (\frac{1}{2}+\Omega(\eps^2))\cdot \opt$.

    Finally, it remains to show that the best of $(\tilde{C},V\setminus \tilde{C})$ and $(\tilde{H},\tilde{L})$ cuts is a $(\frac{1}{2}+\frac{\eps^2}{16})$-approximation for the MAX-CUT value of $G$. This directly follows from \cref{thm:best-approx} (with failure probability $\frac{\delta}{4}$), by substituting $H$ and $L$ with ${\tilde{H}}$ and ${\tilde{L}}$, respectively.
    Together with \cref{lem:high-deg-reservoir}, \cref{lem:error-bound-arbitrary}, and applying union bound, this concludes the proof.
\end{proof}

\subsection{Extension to Dynamic Streams}\label{sec:dynamic}
\paragraph{Overview of the Algorithm.} 
Our algorithm in dynamic streams is basically similar to \cref{alg:arbitrary-order}. The only difference is that
to compute a small set $\tilde{H}$ at the end of the stream such that $\tilde{H} \supseteq H$, instead of reservoir sampling, we need to use $\ell_0$-sampling~\cite{jowhari2011tight}.

Next, we prove the approximation guarantee and the space complexity of our algorithm for dynamic streams. The algorithm is formally given in Algorithm~\ref{alg:dynamic}.

\begin{algorithm}[h!]
\begin{algorithmic}[1]
\Require Graph $G$ as a dynamic stream of edges, an $\eps$-accurate oracle $\mathcal{O} \rightarrow \{-1, 1\}$
\Ensure Estimate of the MAX-CUT value of $G$
\LeftComment{\textbf{Preprocessing phase}}
\State Initialize $e(V^+,V^-) \leftarrow 0$.
\State Set $h_1,\dots,h_k:[n] \rightarrow [w]$ be $2$-wise independent hash functions, with $k=\lceil \frac{e}{\eps^{7}\delta^3}\rceil$ and $w = \lceil \ln \frac{8\beta}{\eps^4\delta^4}\rceil$ (where $\beta$ is a sufficiently large universal constant).
\State Initialize $\CM[V^+]$ and $\CM[V^-]$ to zero. 
\LeftComment{\textbf{Streaming phase}}
\ForEach{item $((u,v), \Delta_{(u,v)})$ in the stream}
\State $Y_u = \mathcal{O}(u), Y_v = \mathcal{O}(v)$
\IIf{$Y_u \neq Y_v$} $e(V^+,V^-) \leftarrow e(V^+,V^-) + 1$
\ForEach{$\ell \in [k]$}
\State $\CM[V^{\mathcal{O}(u)}][\ell,h_\ell(v)] \leftarrow \CM[V^{\mathcal{O}(u)}][\ell,h_\ell(v)] + \Delta_{(u,v)}$
\State $\CM[V^{\mathcal{O}(v)}][\ell,h_\ell(u)] \leftarrow \CM[V^{\mathcal{O}(v)}][\ell,h_\ell(u)] + \Delta_{(u,v)}$
\EndFor
\EndFor
\State In parallel, maintain $\frac{\beta}{\delta^3\eps^4}$ $\ell_0$-samplers (with failure probability $\delta' = \frac{\delta^4\eps^4}{8\beta}$) to obtain $\frac{\beta}{\delta^3\eps^4}$ different edges $F$. \phantomsection \label{line:ell-0-sampling}

\LeftComment{\textbf{Postprocessing phase}}
\ForEach{$v\in V$}
\State $f_v^+ \leftarrow \min_{\ell\in [k]} \CM[V^+][\ell,h_\ell (v)]$
\State $f_v^- \leftarrow \min_{\ell\in [k]} \CM[V^-][\ell,h_\ell (v)]$
\EndFor
\label{line:tilde-H-dynamic}
\State $\tilde{H} := \bigcup_{(u,v)\in F} \{u,v\}$
\State $\tilde{H}^+ \leftarrow \{v\in \tilde{H}: Y_v = 1\}, \tilde{H}^- \leftarrow V \setminus \tilde{H}^+$
\State $\tilde{e}(\tilde{L}^+,\tilde{L}^-) \leftarrow e(V^+,V^-)  - \sum_{v\in \tilde{H}^+} f_v^- - \sum_{v\in \tilde{H}^-} f_v^+$
\State $\mathrm{ALG}_1 \leftarrow \tilde{e}(\tilde{L}^+,\tilde{L}^-) + \sum_{v\in \tilde{H}} \max\{f_v^-,f_v^+\}$
\State $\mathrm{ALG}_2 \leftarrow \sum_{v\in \tilde{H}} (f_v^- + f_v^+)$
\State \textbf{return} $\max\{\mathrm{ALG}_1,\mathrm{ALG}_2\}$

\end{algorithmic}
\caption{Estimating the MAX-CUT value in dynamic streams}
\label{alg:dynamic}
\end{algorithm}

\begin{theorem}
\label{thm:dynamic}
Let $\eps\in (0,\frac{1}{2}]$ and $\delta \in (0,1)$.
Given an $\epsilon$-accurate oracle $\mathcal{O}$ for the MAX-CUT of $G$, 
there exists a single-pass $(\frac{1}{2} + \Omega(\eps^2))$-approximation  algorithm for estimating the MAX-CUT value of $G$ in the dynamic streams. The algorithm uses $O( \frac{\log^2 n}{\eps^7\delta^3 }  \log \frac{1}{\eps\delta} )$ bits of space. 
The approximation holds with probability at least $1-\delta$.
\end{theorem}
Note that \cref{thm:main_introduction} follows from \cref{thm:arbitrary-order} and \cref{thm:dynamic} by setting $\delta=1/3$. 
\paragraph{Proof Sketch of \cref{thm:dynamic}.}
The proof follows a similar structure to that of \cref{thm:arbitrary-order}. First, we show that $\ell_0$-sampling in dynamic streams can effectively replace reservoir sampling in insertion-only streams, allowing us to retrieve a set $\tilde{H} \supseteq H$ with high probability by the end of the stream. 

Next, we show that for all vertices $v \in \tilde{H}$, our estimates of the number of incident edges from $v$ to the sets $V^+$ and $V^-$ has additive error at most $\text{poly}(\epsilon, \delta) \cdot m$, with high probability. 

Finally, using arguments similar to those in the proof of \cref{thm:arbitrary-order}, we can show that \cref{alg:dynamic} provides a $(\frac{1}{2} + \Omega(\epsilon^2))$-approximation for the MAX-CUT value of $G$ with probability $1 - \delta$. For completeness, we provide the detailed proof in \cref{sec:proof-dynamic}.

\section{Conclusion}
We present the first learning-augmented algorithm for estimating the value of MAX-CUT in the streaming setting by leveraging \(\epsilon\)-accurate predictions. Specifically, we provide a single-pass streaming algorithm that achieves a \((1/2 + \Omega(\epsilon^2))\)-approximation for estimating the MAX-CUT value of a graph in insertion-only (respectively, fully dynamic) streams using \(\text{poly}(1/\epsilon)\) (respectively, \(\text{poly}(1/\epsilon, \log n)\)) words of space. This result contrasts with the lower bound in the classical streaming setting (without predictions), where any (randomized) single-pass streaming algorithm that achieves an approximation ratio of at least \(1/2 + \epsilon\) requires \(\Omega(n / 2^{\text{poly}(1/\epsilon)})\) space.

Our work leaves several questions for further research. For example, it would be interesting to extend our algorithm from unweighted to weighted graphs. Currently, our algorithm for handling the low-degree part does not seem to work for weighted graphs, as an edge with a heavy weight but incorrectly predicted signs can significantly affect our estimator. Another open question is to establish lower bounds on the trade-offs between approximation ratio and space complexity when the streaming algorithm is provided with predictions.

\section*{Acknowledgment}

The work was conducted in part while Pan Peng and Ali Vakilian were long-term visitors at the Simons Institute for the Theory of Computing as part of the Sublinear Algorithms program.
Yinhao Dong and Pan Peng are supported in part by NSFC grant 62272431. 

\bibliographystyle{alpha}
\bibliography{main}

\appendix
\section{Proof of \cref{thm:dynamic}}
\label{sec:proof-dynamic}
\begin{lemma}
\label{lem:ell-0-sampling}
Let $\eps \in (0,\frac{1}{2}]$ and $\delta \in (0,1)$. With probability at least $1-\frac{\delta}{4}$, Line~\ref{line:ell-0-sampling} of \cref{alg:dynamic} returns a set $F$ which contains $\frac{\beta}{\delta^3\eps^4}$ different edges. 
\end{lemma}
\begin{proof}
    The error is introduced by the $\ell_0$-samplers and the probability that sampling $\frac{\beta}{\delta^3\eps^4}$ uniform random edges in parallel does not yield $\frac{\beta}{\delta^3\eps^4}$ different edges.

    By \cref{thm:l0-sampler}, the $\ell_0$-sampler fails with probability at most $\delta'$. Since we maintain $\frac{\beta}{\delta^3\eps^4}$ $\ell_0$-samplers, by union bound, the failure probability introduced by $\ell_0$-samplers is at most $\frac{\beta}{\delta^3\eps^4} \cdot \delta'$.

    Suppose that all the $\ell_0$-samplers succeed, i.e., $|F| = \frac{\beta}{\delta^3\eps^4}$.
    Now we bound the probability that there exist two identical edges in $F$. Consider two fixed edges $e_1,e_2 \in F$. Then $\Pr[e_1 = e_2] = \frac{1}{m}$. By union bound, the probability that there exists at least one pair of identical edges is at most $\binom{\beta / \delta^3\eps^4}{2}\cdot \frac{1}{m} \leq \frac{\beta^2}{\delta^6\eps^{8}m}$.

    Therefore, the probability that $F$ contains $\frac{\beta}{\delta^3\eps^4}$ different edges is at least $1-\frac{\beta}{\delta^3\eps^4}\cdot \delta' - \frac{\beta^2}{\delta^6\eps^{8}m} \geq 1-\frac{\delta}{8}-\frac{\delta}{8}\geq 1-\frac{\delta}{4}$ (since $m=\Omega(\eps^{-11}\delta^{-7})$).
\end{proof}

\begin{lemma}
    \label{lem:high-deg-dynamic}
     Let $\delta \in (0,1)$.
     By the end of the stream, $\tilde{H}$ contains all high-degree vertices in $G$ with probability at least $1-\frac{\delta}{4}$.
\end{lemma}
\begin{proof}
    Suppose that $v$ is a high-degree vertex in $G$, i.e., $\deg(v) \geq \frac{\eps^2 m}{c}$. 
    For each edge $e_i$ incident to $v$, define an indicator random variable 
    \begin{equation*}
            X_{i}= \begin{cases}1, & \text { if } e_i \in F \\ 0, & \text { otherwise }\end{cases}
    \end{equation*}
    So, $\E[X_i] = \Pr[X_i=1] = 1-(1-\frac{1}{m})^{\beta  / \delta^3\eps^4} \geq 1-\exp{(-\frac{\beta}{\delta^3\eps^4 m})}\geq \frac{\beta}{\delta^3\eps^4 m} - \frac{\beta^2}{2\delta^6\eps^{8}m^2} \geq \frac{\beta}{5\delta^3\eps^4 m}$ for sufficiently large $m$. We define $X:= \sum_{i\in [\deg(v)]} X_i$ to denote the number of edges incident to $v$ that are sampled by the end of the stream. Then, $\E[X] = \sum_{i\in [\deg(v)]} \E[X_i] \geq \frac{\beta \deg(v)}{5\delta^3\eps^4 m}$. 

    \medskip
    Next, we bound $\Var[X]$. For any $i\in [\deg(v)]$, $\Var[X_i] = \E[X^2_i] - (\E[X_i])^2 = \E[X_i] - (\E[X_i])^2 \leq \frac{\beta}{5\delta^3\eps^4 m} - (\frac{\beta}{5\delta^3\eps^4 m})^2 = \frac{\beta}{5\delta^3\eps^4 m}-\frac{\beta^2}{25\delta^6\eps^{8} m^2}$.
    Since $(X_i)_{i\in [\deg(v)]}$ are negatively correlated, $\Var[X] \le \sum_{i\in [\deg(v)]}\Var[X_i] = \deg(v) \cdot \left( \frac{\beta}{5\delta^3\eps^4 m}-\frac{\beta^2}{25\delta^6\eps^{8} m^2}\right) \leq \frac{\beta\deg(v)}{5\delta^3\eps^4 m}$.
    Therefore, by Chebyshev's inequality, 
    \begin{align*}
        \Pr[v\notin \tilde{H}] = \Pr[X\leq 0] 
        &\leq \Pr\left[X \leq \E[X]- \frac{\beta \deg(v)}{5\delta^3\eps^4 m}\right]\\
        &\leq \frac{\Var[X]\cdot 25\delta^6\eps^{8}m^2}{\beta^2 \deg^2(v)}\\
        &\leq \frac{5\delta^3\eps^4 m}{\beta\cdot \deg(v)} &&\rhd \Var[X] \leq \frac{\beta\deg(v)}{5\delta^3\eps^4 m}\\
        &\leq \frac{5\delta^3\eps^2 c}{\beta}. &&\rhd \deg(v)\geq \frac{\eps^2 m}{c}
    \end{align*}
    By our definition of high-degree vertices, there are at most $\frac{2m}{\eps^2 m / c} =\frac{2c}{\eps^2}$ high-degree vertices in $G$.
    By union bound, with probability at least $1-\frac{2c}{\eps^2}\cdot\frac{5\delta^3\eps^2 c}{\beta} = 1-\frac{10c^2\delta^3}{\beta} \geq 1-\frac{\delta}{4}$ (since $c=\frac{80}{\delta}$ and $\beta$ sufficiently large), $\tilde{H}$ contains all high-degree vertices in $G$.
\end{proof}

\begin{lemma}
\label{lem:error-bound-CM}
    Let $\eps \in (0,\frac{1}{2}]$ and $\delta \in (0,1)$.
    For all vertices $v\in \tilde{H}$, (1) $f_v^+ \geq e(v,V^+)$ and $f_v^- \geq e(v,V^-)$. (2) with probability at least $1-\frac{\delta}{4}$, $f_v^+ \leq e(v,V^+) + 2\eps^{7} \delta^3 m$ and $f_v^- \leq e(v,V^-) + 2\eps^{7} \delta^3 m$.
\end{lemma}
\begin{proof}
    (1) By \cref{thm:CM}, we have $f_v^+ \geq e(v,V^+)$ and $f_v^- \geq e(v,V^-)$ for all $v \in V$.

    (2) By \cref{thm:CM}, we have $f_v^+ \leq e(v,V^+) + 2\eps^{7} \delta^3 m$ and $f_v^- \leq e(v,V^-) + 2\eps^{7}\delta^3 m$ with probability at least $1-\frac{\eps^4\delta^4}{8\beta}$ for any $v \in V$. Since $|\tilde{H}|\leq \frac{2\beta}{\delta^3\eps^4}$, by union bound, with probability at least $1-\frac{\eps^4\delta^4}{8\beta} \cdot \frac{2\beta}{\delta^3\eps^4} = 1-\frac{\delta}{4}$, we have $f_v^+ \leq e(v,V^+) + 2\eps^{7}\delta^3 m$ and $f_v^- \leq e(v,V^-) + 2\eps^{7}\delta^3 m$ for all $v\in \tilde{H}$.
\end{proof}

\begin{proof}[Proof of \cref{thm:dynamic}]
    In \cref{alg:dynamic}, we use CountMin sketch with $k=\lceil \frac{e}{\eps^{7} \delta^3}\rceil$ and $w = \lceil \ln \frac{8\beta}{\eps^4\delta^4}\rceil$ to estimate the number of incident edges of high-degree vertices to with respect to $V^+$ and $V^-$, which takes $O(\frac{1}{\eps^{7}\delta^3}\ln\frac{1}{\eps^4 \delta^4}) = O(\frac{1}{\eps^{7}\delta^3}\ln\frac{1}{\eps \delta})$ words of space, by \cref{thm:CM}. 
    Also, we maintain $O(\frac{1}{\delta^3\eps^4})$ $\ell_0$-samplers, which takes $O(\frac{\log^2 n}{\delta^3 \eps^4}  \log \frac{1}{\delta^4 \eps^4}) = O(\frac{\log^2 n}{\delta^3 \eps^4}  \log \frac{1}{\delta \eps})$ bits of space, by \cref{thm:l0-sampler}.
    Therefore, the total space complexity of \cref{alg:dynamic} is $O( \frac{\log^2 n}{\eps^7\delta^3 }  \log \frac{1}{\eps\delta} )$ bits. 
    
    Next, we show the approximation guarantee of \cref{alg:dynamic}.
    Similar to arbitrary order streams, we cannot detect $H$ by the end of stream. Instead, we have $\tilde{H}$.
    Let $A_1 := e(\tilde{L}^+,\tilde{L}^-) + \sum_{v\in \tilde{H}} \max \{e(v,\tilde{L}^-),e(v,\tilde{L}^+) \}$ denote the value of the cut $(\tilde{C},V\setminus \tilde{C})$ obtained by running \cref{alg:low-deg} on $G[\tilde{L}]$ and then assigning the vertices in $\tilde{H}$ in a greedy manner.
    Let $A_2 := \sum_{v\in \tilde{H}} (e(v,\tilde{L}^-) + e(v,\tilde{L}^+))$ denote the value of the cut $(\tilde{H},\tilde{L})$. Suppose that the best of $(\tilde{C},V\setminus \tilde{C})$ and $(\tilde{H},\tilde{L})$ cuts is a $(\frac{1}{2}+\frac{\eps^2}{16})$-approximation for the MAX-CUT value of $G$.
    In the following, we show that $\mathrm{ALG}_1$ and $\mathrm{ALG}_2$ are good approximations of $A_1$ and $A_2$, respectively. Therefore, $\max\{\mathrm{ALG}_1,\mathrm{ALG}_2\}$ returned by \cref{alg:dynamic} is also a strictly better than $\frac{1}{2}$-approximation for the MAX-CUT value of $G$.

    Since $V^+ = \tilde{H}^+ \cup \tilde{L}^+$ and $V^- = \tilde{H}^- \cup \tilde{L}^-$, we have $e(v,V^+) = e(v,\tilde{H}^+) + e(v,\tilde{L}^+)$ and $e(v,V^-) = e(v,\tilde{H}^-) + e(v,\tilde{L}^-)$ for any vertex $v\in V$. By \cref{lem:error-bound-CM}, with probability at least $1-\frac{\delta}{4}$, for all vertices $v\in \tilde{H}$, we have 
    \begin{align*}
        e(v,\tilde{H}^+) + e(v,\tilde{L}^+) \leq f_v^+ \leq e(v,\tilde{H}^+) + e(v,\tilde{L}^+) + 2\eps^{7}\delta^3 m,\\
        e(v,\tilde{H}^-) + e(v,\tilde{L}^-) \leq f_v^- \leq e(v,\tilde{H}^-) + e(v,\tilde{L}^-) + 2\eps^{7}\delta^3m.
    \end{align*}
    
    Since $|\tilde{H}|\leq \frac{2\beta}{\delta^3\eps^4}$,
    we have $0 \leq e(v,\tilde{H}^+),e(v,\tilde{H}^-) \leq \frac{2\beta}{\delta^3\eps^{4}}$. Therefore, we have
    \begin{align*}
        e(v,\tilde{L}^+) \leq f_v^+ \leq \frac{2\beta}{\delta^3\eps^{4}} +2\eps^{7}\delta^3 m+ e(v,\tilde{L}^+),\\
        e(v,\tilde{L}^-) \leq f_v^- \leq \frac{2\beta}{\delta^3\eps^{4}} +2\eps^{7}\delta^3 m + e(v,\tilde{L}^-).
    \end{align*}
    Then we have
    \begin{align*}
        0 \leq f_v^+ - e(v,\tilde{L}^+)\leq \frac{2\beta}{\delta^3\eps^{4}} +2\eps^{7}\delta^3 m && \text{and} && 0 \leq f_v^- - e(v,\tilde{L}^-)\leq \frac{2\beta}{\delta^3\eps^{4}} +2\eps^{7}\delta^3 m,
    \end{align*}
    and
    \begin{align*}
        \max\{e(v,\tilde{L}^+),e(v,\tilde{L}^-) \}  \leq \max\{ f_v^+,f_v^-\} \leq \frac{2\beta}{\delta^3\eps^{4}} + 2\eps^7 
 \delta^3 m + \max\{e(v,\tilde{L}^+),e(v,\tilde{L}^-) \}.
    \end{align*}

    Note that 
    \begin{align*}
        e(\tilde{L}^+,\tilde{L}^-) &= e(V^+,V^-) - e(\tilde{H}^+,\tilde{H}^-) - \sum_{v\in \tilde{H}^+}e(v,\tilde{L}^-) - \sum_{v\in \tilde{H}^-}e(v,\tilde{L}^+),\\
        \tilde{e}(\tilde{L}^+,\tilde{L}^-) &= e(V^+,V^-)  - \sum_{v\in \tilde{H}^+} f_v^- - \sum_{v\in \tilde{H}^-} f_v^+.
    \end{align*}
    We have 
    \begin{align*}
        e(\tilde{L}^+,\tilde{L}^-) - \tilde{e}(\tilde{L}^+,\tilde{L}^-) &= \sum_{v\in \tilde{H}^+} (f_v^- - e(v,\tilde{L}^-)) + \sum_{v\in \tilde{H}^-} (f_v^+ - e(v,\tilde{L}^+)) - e(\tilde{H}^+,\tilde{H}^-)\\
        &\leq \frac{2\beta}{\delta^3\eps^{4}}\cdot \left(\frac{2\beta}{\delta^3\eps^{4}} + 2\eps^7 \delta^3 m\right) + \frac{2\beta}{\delta^3\eps^{4}}\cdot \left(\frac{2\beta}{\delta^3\eps^{4}} + 2\eps^7  \delta^3 m\right) - 0
        = \frac{8\beta^2}{\delta^6\eps^8} + 8\beta \eps^3 m.
    \end{align*}

    Therefore,
    \begin{align*}
        A_1 - \mathrm{ALG}_1 &= (e(\tilde{L}^+,\tilde{L}^-) - \tilde{e}(\tilde{L}^+,\tilde{L}^-)) + \sum_{v\in\tilde{H}} (\max\{e(v,\tilde{L}^-), e(v,\tilde{L}^+)\} - \max\{f_v^-,f_v^+\})\\
        &\leq \left(\frac{8\beta^2}{\delta^6\eps^8} + 8\beta \eps^3 m \right)+ \frac{2\beta}{\delta^3\eps^4} \cdot 0
        = \frac{8\beta^2}{\delta^6\eps^8} + 8\beta \eps^3 m.
    \end{align*}
    So, $\mathrm{ALG}_1 \geq A_1 - (\frac{8\beta^2}{\delta^6\eps^8} + 8\beta \eps^3 m) = A_1 - \Theta(\eps^3 m)$. 

    Similarly,
    \begin{align*}
        A_2 - \mathrm{ALG}_2 = \sum_{v\in \tilde{H}} ((  e(v,\tilde{L}^-) - f_v^-) + (  e(v,\tilde{L}^+) - f_v^+))
        \leq \frac{2\beta}{\delta^3 \eps^4} \cdot (0+0) = 0.
    \end{align*}
    So, $\mathrm{ALG}_2 \geq A_2$. 
    
    Since we assume that $\max\{A_1,A_2\}$ is a $(\frac{1}{2}+\frac{\eps^2}{16})$-approximation for the MAX-CUT value of $G$, we have $\mathrm{ALG}:= \max\{\mathrm{ALG}_1, \mathrm{ALG}_2\} \geq \max\{A_1,A_2\}- \Theta(\eps^3 m)  \geq (\frac{1}{2}+\frac{\eps^2}{16})\cdot \opt - \Theta(\eps^3 \cdot \opt) = (\frac{1}{2}+\Omega(\eps^2))\cdot \opt$.

    Finally, it remains to show that the best of $(\tilde{C},V\setminus \tilde{C})$ and $(\tilde{H},\tilde{L})$ cuts is a $(\frac{1}{2}+\frac{\eps^2}{16})$-approximation for the MAX-CUT value of $G$. This directly follows from \cref{thm:best-approx} (with failure probability $\frac{\delta}{4}$), by substituting $H$ and $L$ with ${\tilde{H}}$ and ${\tilde{L}}$, respectively.
    Together with \cref{lem:ell-0-sampling}, \cref{lem:high-deg-dynamic}, \cref{lem:error-bound-CM}, and applying union bound, this concludes the proof.
\end{proof}

\section{Constant Query Complexity in Random Order Streams}\label{sec:constantquery}

In this section, we show that in the case of random-order streams, it suffices to query the $\eps$-accurate oracle $\mathcal{O}$ for the labels of a constant number of vertices. For simplicity, we provide only a sketch of the main ideas and omit the formal proof. We will utilize the following lemma.
\begin{lemma}[Lemma 2.2 in \cite{bhattacharyya22property}]
\label{lem:sum}
	Let $a,b\in\mathbb{R}$ with $a<b$, $n\in\mathbb{N}$ and $x\in[a,b]^n$.
	Let $\Sigma = \sum_{i=1}^n x_i$.
	For any $\eta,\delta\in(0,1)$, there exists an algorithm which
	samples a set $T$ of $t=O(\eta^{-2} \log \delta^{-1})$ indices from $[n]$ and returns an
	estimate $\tilde{\Sigma}=\frac{n}{t}\sum_{i\in T}x_i$ such that $|\Sigma - \tilde{\Sigma}| \leq \eta(b-a) n$ with
	probability at least $1-\delta$.
\end{lemma}

Similar to \cref{alg:rand-order}, we store the first $\poly(1/\eps,1/\delta)$ edges in the random order stream and use them to identify a set $\tilde{H}$ of size $\poly(1/\eps,1/\delta)$ that contains all high-degree vertices $H$ with high probability as in \cref{lem:high-deg-sample}. For this part, we do not need any information on the labels of vertices provided by the oracle $\mathcal{O}$.
Recall that the algorithm considers two candidate cuts and returns the one with the larger size. Let $\tilde{L}:= V\setminus \tilde{H}$ and $\tilde{S}:= \tilde{H}\setminus H$.
The first candidate is obtained by performing the greedy extension of $(\tilde{L}^+ \cup \tilde{S}^+, \tilde{L}^- \cup \tilde{S}^-)$ using $H$. The second candidate is simply the cut $(H,L)$. Formally, the sizes of these two cuts are given as follows:
\begin{align*}
    \mathrm{ALG}_1 &= e(L^+,L^-) + \sum_{v\in H} \max \{e(v,L^-),e(v,L^+)\}\\
    &= e(\tilde{L}^+,\tilde{L}^-) + e(\tilde{S}^+,\tilde{S}^-) + \sum_{v\in \tilde{S}^+} e(v,\tilde{L}^-)+\sum_{v\in \tilde{S}^-} e(v,\tilde{L}^+) +  \sum_{v\in H} \max \{e(v,L^-),e(v,L^+)\},\\
    \mathrm{ALG}_2 &= e(H,L) = \sum_{v\in H} e(v,L) = \sum_{v\in H} (e(v,\tilde{L}) + e(v,\tilde{S})).
\end{align*}

Observe that, to compute the second cut size $\mathrm{ALG}_2$, there is no need to query the oracle $\mathcal{O}$. It suffices to count $e(v,\tilde{L})$ for each vertex $v \in \tilde{H}$ using the remaining edges (after the first $\poly(1/\eps,1/\delta)$ edges) during the stream. At the end of the stream, we can retrieve $H$ from $\tilde{H}$ and compute $\mathrm{ALG}_2$ exactly.

Next, we focus on estimating $\mathrm{ALG}_1$ by querying the oracle $\mathcal{O}$ a constant number of times. Specifically, we decompose $\mathrm{ALG}_1$ into four parts and estimate each part respectively.

\begin{enumerate}[label=(\textsf{\Roman*}),leftmargin = *]
    \item\label{case:edges-tilde-L} $e(\tilde{L}^+,\tilde{L}^-)$. 
    During the stream (after the first $\poly(1/\eps,1/\delta)$ edges), we employ reservoir sampling to sample $O(\eps^{-4}\log \delta^{-1})$ edges $T$ uniformly at random from $E(G[\tilde{L}])$, the set of edges with both endpoints in $\tilde{L}$. Let $U:= \bigcup_{(u,v)\in T}\{u,v\}$. By \cref{lem:sum}, we query $\mathcal{O}$ for the vertices in $U$ and use $\frac{e(G[\tilde{L}])}{|T|}\cdot e(U^+,U^-)$ to estimate $e(\tilde{L}^+,\tilde{L}^-)$, with an additive error of $\frac{\eps^2}{64} \cdot e(G[\tilde{L}])$.
    
    \item\label{case:edges-tilde-S} $e(\tilde{S}^+,\tilde{S}^-)$. Since we store all  edges with both endpoints in $\tilde{H}$ during the stream and $|\tilde{H}| = \poly(1/\eps,1/\delta)$, we can compute $e(\tilde{S}^+,\tilde{S}^-)$ exactly by querying $\mathcal{O}$ a constant number of times.

    \item\label{case:edges-cross} $\sum_{v\in \tilde{S}^+} e(v,\tilde{L}^-)+\sum_{v\in \tilde{S}^-} e(v,\tilde{L}^+)$.
    For each vertex $v \in \tilde{H}$, we use reservoir sampling to sample $O(\eps^{-4}\log \delta^{-1})$ edges $T_v$ uniformly at random from $E(v,\tilde{L})$, the set of edges with one endpoint at $v$ and the other in $\tilde{L}$. 
    Let $U_v := \bigcup_{(u,v)\in T_v} \{u\}$. By \cref{lem:sum}, we query $\mathcal{O}$ for the vertices in $U_v$ and use $\frac{e(v,\tilde{L})}{|T_v|}\cdot e(v,U_v^+)$ to estimate $e(v,\tilde{L}^+)$ 
    and use $\frac{e(v,\tilde{L})}{|T_v|}\cdot e(v,U_v^-)$ to estimate $e(v,\tilde{L}^-)$, both with an additive error of $\frac{\eps^2}{64} \cdot e(v,\tilde{L})$.
    Therefore, we can estimate $\sum_{v\in \tilde{S}^+} e(v,\tilde{L}^-)+\sum_{v\in \tilde{S}^-} e(v,\tilde{L}^+)$ with an additive error of $\frac{\eps^2}{64}\cdot \sum_{v\in \tilde{S}} e(v,\tilde{L}) = \frac{\eps^2}{64} \cdot e(\tilde{S},\tilde{L})$.
    
    \item\label{case:greedy} $\sum_{v\in H} \max \{e(v,L^-),e(v,L^+)\}$.
    Since $L^- = \tilde{S}^- \cup \tilde{L}^-$ and $L^+ = \tilde{S}^+ \cup \tilde{L}^+$, we have $e(v,L^-) = e(v,\tilde{S}^-)+e(v,\tilde{L}^-)$ and $e(v,L^+) = e(v,\tilde{S}^+)+e(v,\tilde{L}^+)$, for each vertex $v\in H$. Similar to \ref{case:edges-tilde-S} and \ref{case:edges-cross}, we can estimate $\sum_{v\in H} \max \{e(v,L^-),e(v,L^+)\}$ with an additive error of $\frac{\eps^2}{64} \cdot \sum_{v\in H}e(v,\tilde{L}) = \frac{\eps^2}{64} \cdot e(H,\tilde{L})$.
\end{enumerate}

Let $\widetilde{\mathrm{ALG}}_1$ denote our estimator for $\mathrm{ALG}_1$. Then we have $|\widetilde{\mathrm{ALG}}_1 - \mathrm{ALG}_1| \leq \frac{\eps^2}{64} \cdot (e(G[\tilde{L}]) + e(\tilde{S},\tilde{L}) + e(H,\tilde{L})) \le \frac{\eps^2}{64} \cdot (e(G[L]) + e(H,L))\le \frac{\eps^2}{64} m$.

If $\widetilde{\mathrm{ALG}}_2:=\mathrm{ALG}_2 = e(H,L) \ge (\frac{1}{2}+\eps^2)\cdot \mathrm{OPT}$, then we are done. Hence, without loss of generality, we assume that $\mathrm{ALG}_2 = e(H,L) < (\frac{1}{2}+\eps^2)\cdot \mathrm{OPT}$. By the proof of \cref{thm:best-approx}, we have $\mathrm{ALG}_1 \ge (\frac{1}{2} + \frac{\eps^2}{16})\cdot \mathrm{OPT}$. 
Therefore, $|\widetilde{\mathrm{ALG}}_1 - \mathrm{ALG}_1| \le \frac{\eps^2}{64} m\leq \frac{\eps^2}{32} \mathrm{OPT}$.
Thus, $\widetilde{\mathrm{ALG}}_1\geq  (\frac{1}{2} + \frac{\eps^2}{16})\mathrm{OPT}-\frac{\eps^2}{32}\mathrm{OPT}=(\frac{1}{2} + \Omega{(\eps^2)})\mathrm{OPT}$. Therefore, $\max\{\widetilde{\mathrm{ALG}}_1,\widetilde{\mathrm{ALG}}_2\}\geq (\frac{1}{2}+\Omega(\eps^2))\mathrm{OPT}$.

\paragraph{Query Complexity.} In the above analysis,  
the total query complexity for the $\eps$-accurate oracle $\mathcal{O}$ is $\poly(1/\eps,1/\delta)$, which is independent of $n$.

\end{document}

%% file: math_commands.tex
\usepackage{amsmath,amsfonts,bm}

\def\eqref#1{equation~\ref{#1}}

\def\1{\bm{1}}

\def\eps{{\epsilon}}

\DeclareMathAlphabet{\mathsfit}{\encodingdefault}{\sfdefault}{m}{sl}
\SetMathAlphabet{\mathsfit}{bold}{\encodingdefault}{\sfdefault}{bx}{n}

\def\sR{{\mathbb{R}}}

\newcommand{\E}{\mathbb{E}}

\newcommand{\Var}{\mathrm{Var}}

\newcommand{\Cov}{\mathrm{Cov}}


%% file: arxiv-version.bbl
\newcommand{\etalchar}[1]{$^{#1}$}
\begin{thebibliography}{BCAGL23}

\bibitem[ACE{\etalchar{+}}23]{antoniadis2023online}
Antonios Antoniadis, Christian Coester, Marek Eli{\'{a}}s, Adam Polak, and Bertrand Simon.
\newblock Online metric algorithms with untrusted predictions.
\newblock {\em {ACM} Trans. Algorithms}, 19(2):19:1--19:34, 2023.

\bibitem[ACN{\etalchar{+}}23]{aamand2023improved}
Anders Aamand, Justin Chen, Huy Nguyen, Sandeep Silwal, and Ali Vakilian.
\newblock Improved frequency estimation algorithms with and without predictions.
\newblock In {\em Advances in Neural Information Processing Systems 36: Annual Conference on Neural Information Processing Systems (NeurIPS)}, 2023.

\bibitem[ADJ{\etalchar{+}}20]{angelopoulos2020online}
Spyros Angelopoulos, Christoph D{\"u}rr, Shendan Jin, Shahin Kamali, and Marc Renault.
\newblock Online computation with untrusted advice.
\newblock In {\em 11th Innovations in Theoretical Computer Science Conference (ITCS)}, volume 151 of {\em LIPIcs}, pages 52:1--52:15, 2020.

\bibitem[AGKK20]{antoniadis2020secretary}
Antonios Antoniadis, Themis Gouleakis, Pieter Kleer, and Pavel Kolev.
\newblock Secretary and online matching problems with machine learned advice.
\newblock In {\em Advances in Neural Information Processing Systems 33: Annual Conference on Neural Information Processing Systems (NeurIPS)}, 2020.

\bibitem[AGM12]{ahn2012graph}
Kook~Jin Ahn, Sudipto Guha, and Andrew McGregor.
\newblock Graph sketches: sparsification, spanners, and subgraphs.
\newblock In {\em Proceedings of the 31st ACM SIGMOD-SIGACT-SIGAI symposium on Principles of Database Systems (PODS)}, pages 5--14, 2012.

\bibitem[AIV19]{aamand2019frequency}
Anders Aamand, Piotr Indyk, and Ali Vakilian.
\newblock Frequency estimation algorithms under zipfian distribution.
\newblock {\em arXiv preprint arXiv:1908.05198}, 2019.

\bibitem[Bal20]{balcan2020data}
Maria-Florina Balcan.
\newblock Data-driven algorithm design.
\newblock {\em arXiv preprint arXiv:2011.07177}, 2020.

\bibitem[BCAGL23]{banerjee2023graph}
Siddhartha Banerjee, Vincent Cohen-Addad, Anupam Gupta, and Zhouzi Li.
\newblock Graph searching with predictions.
\newblock In {\em 14th Innovations in Theoretical Computer Science Conference (ITCS)}, volume 251 of {\em LIPIcs}, pages 12:1--12:24, 2023.

\bibitem[BDSW24]{braverman2024learning}
Vladimir Braverman, Prathamesh Dharangutte, Vihan Shah, and Chen Wang.
\newblock Learning-augmented maximum independent set.
\newblock In {\em Approximation, Randomization, and Combinatorial Optimization. Algorithms and Techniques ({APPROX/RANDOM})}, volume 317 of {\em LIPIcs}, pages 24:1--24:18, 2024.

\bibitem[BFNP24]{brand2024dynamic}
Jan van~den Brand, Sebastian Forster, Yasamin Nazari, and Adam Polak.
\newblock On dynamic graph algorithms with predictions.
\newblock In {\em Proceedings of the 2024 {ACM-SIAM} Symposium on Discrete Algorithms ({SODA})}, pages 3534--3557, 2024.

\bibitem[BMS20]{bamas2020primal}
Etienne Bamas, Andreas Maggiori, and Ola Svensson.
\newblock The primal-dual method for learning augmented algorithms.
\newblock In {\em Advances in Neural Information Processing Systems 33: Annual Conference on Neural Information Processing Systems (NeurIPS)}, 2020.

\bibitem[BNVW17]{balcan2017learning}
Maria-Florina Balcan, Vaishnavh Nagarajan, Ellen Vitercik, and Colin White.
\newblock Learning-theoretic foundations of algorithm configuration for combinatorial partitioning problems.
\newblock In {\em Conference on Learning Theory}, pages 213--274. PMLR, 2017.

\bibitem[BY22]{bhattacharyya22property}
Arnab Bhattacharyya and Yuichi Yoshida.
\newblock {\em Property Testing - Problems and Techniques}.
\newblock Springer, 2022.

\bibitem[CdG{\etalchar{+}}24]{cohen2024max}
Vincent Cohen{-}Addad, Tommaso d'Orsi, Anupam Gupta, Euiwoong Lee, and Debmalya Panigrahi.
\newblock Learning-augmented approximation algorithms for maximum cut and related problems.
\newblock In {\em Advances in Neural Information Processing Systems 37: Annual Conference on Neural Information Processing Systems (NeurIPS)}, 2024.

\bibitem[CEI{\etalchar{+}}22]{chen2022triangle}
Justin~Y. Chen, Talya Eden, Piotr Indyk, Honghao Lin, Shyam Narayanan, Ronitt Rubinfeld, Sandeep Silwal, Tal Wagner, David~P. Woodruff, and Michael Zhang.
\newblock Triangle and four cycle counting with predictions in graph streams.
\newblock In {\em The Tenth International Conference on Learning Representations ({ICLR})}, 2022.

\bibitem[CM05]{cormode2005improved}
Graham Cormode and Shan Muthukrishnan.
\newblock An improved data stream summary: the count-min sketch and its applications.
\newblock {\em Journal of Algorithms}, 55(1):58--75, 2005.

\bibitem[CSVZ22]{chen2022faster}
Justin Chen, Sandeep Silwal, Ali Vakilian, and Fred Zhang.
\newblock Faster fundamental graph algorithms via learned predictions.
\newblock In {\em International Conference on Machine Learning (ICML)}, volume 162 of {\em Proceedings of Machine Learning Research}, pages 3583--3602, 2022.

\bibitem[DIL{\etalchar{+}}21]{dinitz2021faster}
Michael Dinitz, Sungjin Im, Thomas Lavastida, Benjamin Moseley, and Sergei Vassilvitskii.
\newblock Faster matchings via learned duals.
\newblock In {\em Advances in Neural Information Processing Systems 34: Annual Conference on Neural Information Processing Systems (NeurIPS)}, pages 10393--10406, 2021.

\bibitem[DMVW23]{davies2023predictive}
Sami Davies, Benjamin Moseley, Sergei Vassilvitskii, and Yuyan Wang.
\newblock Predictive flows for faster ford-fulkerson.
\newblock In {\em International Conference on Machine Learning (ICML)}, volume 202 of {\em Proceedings of Machine Learning Research}, pages 7231--7248, 2023.

\bibitem[DTV24]{depavia2024learning}
Adela~F DePavia, Erasmo Tani, and Ali Vakilian.
\newblock Learning-based algorithms for graph searching problems.
\newblock In {\em International Conference on Artificial Intelligence and Statistics (AISTATS)}, volume 238 of {\em Proceedings of Machine Learning Research}, pages 928--936, 2024.

\bibitem[EIN{\etalchar{+}}21]{eden2021learning}
Talya Eden, Piotr Indyk, Shyam Narayanan, Ronitt Rubinfeld, Sandeep Silwal, and Tal Wagner.
\newblock Learning-based support estimation in sublinear time.
\newblock In {\em 9th International Conference on Learning Representations ({ICLR})}, 2021.

\bibitem[FV20]{ferragina2020the}
Paolo Ferragina and Giorgio Vinciguerra.
\newblock The pgm-index: a fully-dynamic compressed learned index with provable worst-case bounds.
\newblock {\em Proc. {VLDB} Endow.}, 13(8):1162--1175, 2020.

\bibitem[GMM25]{ghoshal2025constraint}
Suprovat Ghoshal, Konstantin Makarychev, and Yury Makarychev.
\newblock Constraint satisfaction problems with advice.
\newblock In {\em Proceedings of the 2025 {ACM-SIAM} Symposium on Discrete Algorithms ({SODA})}, 2025.

\bibitem[GW95]{goemans1995improved}
Michel~X Goemans and David~P Williamson.
\newblock Improved approximation algorithms for maximum cut and satisfiability problems using semidefinite programming.
\newblock {\em Journal of the ACM (JACM)}, 42(6):1115--1145, 1995.

\bibitem[HIKV19]{hsu2019learning}
Chen-Yu Hsu, Piotr Indyk, Dina Katabi, and Ali Vakilian.
\newblock Learning-based frequency estimation algorithms.
\newblock In {\em 7th International Conference on Learning Representations ({ICLR})}, 2019.

\bibitem[HSSY24]{henzinger2024complexity}
Monika Henzinger, Barna Saha, Martin~P Seybold, and Christopher Ye.
\newblock On the complexity of algorithms with predictions for dynamic graph problems.
\newblock In {\em 15th Innovations in Theoretical Computer Science Conference (ITCS)}, volume 287 of {\em LIPIcs}, pages 62:1--62:25, 2024.

\bibitem[IKMQP21]{im2021online}
Sungjin Im, Ravi Kumar, Mahshid Montazer~Qaem, and Manish Purohit.
\newblock Online knapsack with frequency predictions.
\newblock In {\em Advances in Neural Information Processing Systems 34: Annual Conference on Neural Information Processing Systems (NeurIPS)}, pages 2733--2743, 2021.

\bibitem[IVY19]{indyk2019learning}
Piotr Indyk, Ali Vakilian, and Yang Yuan.
\newblock Learning-based low-rank approximations.
\newblock In {\em Advances in Neural Information Processing Systems 32: Annual Conference on Neural Information Processing Systems (NeurIPS)}, pages 7400--7410, 2019.

\bibitem[JLL{\etalchar{+}}20]{jiang2020learning}
Tanqiu Jiang, Yi~Li, Honghao Lin, Yisong Ruan, and David~P Woodruff.
\newblock Learning-augmented data stream algorithms.
\newblock In {\em 8th International Conference on Learning Representations ({ICLR})}, 2020.

\bibitem[JST11]{jowhari2011tight}
Hossein Jowhari, Mert Sa{\u{g}}lam, and G{\'a}bor Tardos.
\newblock Tight bounds for lp samplers, finding duplicates in streams, and related problems.
\newblock In {\em Proceedings of the 30th ACM SIGMOD-SIGACT-SIGART symposium on Principles of database systems (PODS)}, pages 49--58, 2011.

\bibitem[KK15]{kogan2015sketching}
Dmitry Kogan and Robert Krauthgamer.
\newblock Sketching cuts in graphs and hypergraphs.
\newblock In {\em Proceedings of the 2015 Conference on Innovations in Theoretical Computer Science (ITCS)}, pages 367--376, 2015.

\bibitem[KK19]{KK19}
Michael Kapralov and Dmitry Krachun.
\newblock An optimal space lower bound for approximating {MAX-CUT}.
\newblock In {\em Proceedings of the 51st Annual {ACM} {SIGACT} Symposium on Theory of Computing ({STOC})}, pages 277--288, 2019.

\bibitem[KKMO07]{khot2007optimal}
Subhash Khot, Guy Kindler, Elchanan Mossel, and Ryan O’Donnell.
\newblock Optimal inapproximability results for max-cut and other 2-variable csps?
\newblock {\em SIAM Journal on Computing}, 37(1):319--357, 2007.

\bibitem[KKS15]{kapralov2014streaming}
Michael Kapralov, Sanjeev Khanna, and Madhu Sudan.
\newblock Streaming lower bounds for approximating max-cut.
\newblock In {\em Proceedings of the Twenty-Sixth Annual ACM-SIAM Symposium on Discrete Algorithms (SODA)}, pages 1263--1282, 2015.

\bibitem[KKSV17]{kapralov20171+}
Michael Kapralov, Sanjeev Khanna, Madhu Sudan, and Ameya Velingker.
\newblock (1+ $\omega(1)$)-approximation to max-cut requires linear space.
\newblock In {\em Proceedings of the Twenty-Eighth Annual ACM-SIAM Symposium on Discrete Algorithms (SODA)}, pages 1703--1722, 2017.

\bibitem[KMM{\etalchar{+}}20]{kapralov2020fast}
Michael Kapralov, Aida Mousavifar, Cameron Musco, Christopher Musco, Navid Nouri, Aaron Sidford, and Jakab Tardos.
\newblock Fast and space efficient spectral sparsification in dynamic streams.
\newblock In {\em Proceedings of the 2020 {ACM-SIAM} Symposium on Discrete Algorithms ({SODA})}, pages 1814--1833, 2020.

\bibitem[KPS18]{kumar2018improving}
Ravi Kumar, Manish Purohit, and Zoya Svitkina.
\newblock Improving online algorithms via ml predictions.
\newblock In {\em Advances in Neural Information Processing Systems 31: Annual Conference on Neural Information Processing Systems (NeurIPS)}, pages 9684--9693, 2018.

\bibitem[LLL{\etalchar{+}}23]{li2023learning}
Yi~Li, Honghao Lin, Simin Liu, Ali Vakilian, and David~P. Woodruff.
\newblock Learning the positions in countsketch.
\newblock In {\em The Eleventh International Conference on Learning Representations (ICLR)}, 2023.

\bibitem[LLMV20]{lattanzi2020online}
Silvio Lattanzi, Thomas Lavastida, Benjamin Moseley, and Sergei Vassilvitskii.
\newblock Online scheduling via learned weights.
\newblock In {\em Proceedings of the 2020 {ACM-SIAM} Symposium on Discrete Algorithms ({SODA})}, pages 1859--1877, 2020.

\bibitem[LLW22]{lin2022learning}
Honghao Lin, Tian Luo, and David~P. Woodruff.
\newblock Learning augmented binary search trees.
\newblock In {\em International Conference on Machine Learning ({ICML})}, volume 162 of {\em Proceedings of Machine Learning Research}, pages 13431--13440, 2022.

\bibitem[LS23]{liu2023predicted}
Quanquan~C Liu and Vaidehi Srinivas.
\newblock The predicted-deletion dynamic model: Taking advantage of ml predictions, for free.
\newblock {\em arXiv preprint arXiv:2307.08890}, 2023.

\bibitem[LSV23]{lattanzi2023speeding}
Silvio Lattanzi, Ola Svensson, and Sergei Vassilvitskii.
\newblock Speeding up bellman ford via minimum violation permutations.
\newblock In {\em International Conference on Machine Learning (ICML)}, volume 202 of {\em Proceedings of Machine Learning Research}, pages 18584--18598, 2023.

\bibitem[LV21]{lykouris2021competitive}
Thodoris Lykouris and Sergei Vassilvitskii.
\newblock Competitive caching with machine learned advice.
\newblock {\em J. {ACM}}, 68(4):24:1--24:25, 2021.

\bibitem[Mit18]{mitzenmacher2018model}
Michael Mitzenmacher.
\newblock A model for learned bloom filters and optimizing by sandwiching.
\newblock In {\em Advances in Neural Information Processing Systems 31: Annual Conference on Neural Information Processing Systems (NeurIPS)}, pages 462--471, 2018.

\bibitem[MNS07]{mahdian2007allocating}
Mohammad Mahdian, Hamid Nazerzadeh, and Amin Saberi.
\newblock Allocating online advertisement space with unreliable estimates.
\newblock In {\em Proceedings of the 8th ACM conference on Electronic commerce}, pages 288--294, 2007.

\bibitem[SCI{\etalchar{+}}23]{schiefer2023learned}
Nicholas Schiefer, Justin~Y Chen, Piotr Indyk, Shyam Narayanan, Sandeep Silwal, and Tal Wagner.
\newblock Learned interpolation for better streaming quantile approximation with worst-case guarantees.
\newblock In {\em SIAM Conference on Applied and Computational Discrete Algorithms (ACDA)}, pages 87--97, 2023.

\bibitem[SM23]{sato2023fast}
Atsuki Sato and Yusuke Matsui.
\newblock Fast partitioned learned bloom filter.
\newblock In {\em Advances in Neural Information Processing Systems 36: Annual Conference on Neural Information Processing Systems (NeurIPS)}, 2023.

\bibitem[Vit85]{vitter1985random}
Jeffrey~S Vitter.
\newblock Random sampling with a reservoir.
\newblock {\em ACM Transactions on Mathematical Software (TOMS)}, 11(1):37--57, 1985.

\bibitem[VKMK21]{VaidyaKMK21}
Kapil Vaidya, Eric Knorr, Michael Mitzenmacher, and Tim Kraska.
\newblock Partitioned learned bloom filters.
\newblock In {\em 9th International Conference on Learning Representations (ICLR)}, 2021.

\end{thebibliography}
